%% file: Selfstabilizing.tex
\newif\ifarticle

\documentclass[11pt]{article} \articletrue

\input{head}          
\input{commands}      

\begin{document}
\input{title-authors} 

\ifarticle\else
    \input{Sections/abstract}
\fi

\maketitle

\ifarticle
    \input{Sections/abstract}
\fi



\input{Sections/introduction}

\input{Sections/preliminaries}

\input{Sections/resetting-subprotocol}

\input{Sections/linear-time-linear-state-silent}

\input{Sections/log-time-bounded-state}

\input{Sections/synthetic-coin}

\input{Sections/conclusion}

\section*{Acknowledgement}
We warmly thank anonymous reviewers for their detailed comments, which have improved the paper greatly.
Doty and Severson were supported by NSF award 1900931 and CAREER award 1844976.
Ho-Lin and Hsueh-Ping were supported by MOST (Taiwan) grant number 107-2221-E-002-031-MY3.
Nowak was supported by the CNRS project ABIDE.

\bibliographystyle{plain}
\bibliography{ref}




\end{document}

%% file: head.tex
\usepackage[appendix=inline]{apxproof} \usepackage[inline]{optional}

\ifarticle
    \usepackage{cite}
    \usepackage{hyperref}
    \usepackage{fullpage}
    \usepackage{authblk}
\fi

\usepackage{physics}

\usepackage[
    textsize=scriptsize
]{todonotes} 

\newtheoremrep{theorem}{Theorem}[section]
\newtheoremrep{lemma}[theorem]{Lemma}
\newtheoremrep{corollary}[theorem]{Corollary}
\newtheoremrep{observation}[theorem]{Observation}

\usepackage{amsmath}
\usepackage{amssymb}
\usepackage{amsthm}
\usepackage{forest}

\usepackage{algorithm}
\usepackage{caption}
\usepackage[noend]{algpseudocode}

\algnewcommand{\IfOneLine}[2]{
  \State \algorithmicif\ #1\ \algorithmicthen\ #2}
  
\makeatletter
\newenvironment{protocol}[1][htb]{%
    \renewcommand{\ALG@name}{Protocol}
   \begin{algorithm}[#1]%
  }{\end{algorithm}}
\makeatother

\usepackage{amsmath}
\DeclareMathOperator{\Dist}{Dist}

\usepackage[normalem]{ulem}

\DeclareMathAlphabet{\mathpzc}{OT1}{pzc}{m}{it}

%% file: commands.tex
\renewcommand{\mathtt}[1]{{\text {\tt #1}}}
\renewcommand{\mathbf}[1]{{\text {\bf #1}}}
\renewcommand{\mathsf}[1]{{\text {\sf #1}}}

\def\Ttable{\mathbf{T}}
\def\Agents{\mathcal{A}}
\def\Q{\{0,1\}^{\leq 3\log_2 n}}

\def\true{\mathsf{True}}
\def\false{\mathsf{False}}

\newcommand{\resetprotocol}{\textsc{Reset}}

\newcommand{\OriginalSSLE}{\hyperref[algo:original]{\textsc{Silent-n-state-SSR}}}
\newcommand{\propagateresetprotocol}{\hyperref[algo:propag-reset]{\textsc{Propagate-Reset}}}
\newcommand{\silentlinearTimeStateProtocol}{\hyperref[algo:silent-linear-time-state]{\textsc{Optimal-Silent-SSR}}}
\newcommand{\logTimeProtocol}{\hyperref[algo:lineartime_bounded]{\textsc{Sublinear-Time-SSR}}}
\newcommand{\detectcollision}{\hyperref[algo:detectcollision]{\textsc{Detect-Name-Collision}}}
\newcommand{\pathconsistencyprotocol}{\hyperref[algo:pathconsistency]{\textsc{Check-Path-Consistency}}}

\newcommand{\IE}{\mathbb{E}}
\newcommand{\IP}{\mathbb{P}}
\newcommand{\IN}{\mathbb{N}}

\newcommand{\timer}{\mathtt{timer}}

\newcommand{\name}{\mathtt{name}}

\newcommand{\Met}{\mathtt{roster}}
\newcommand{\level}{\mathtt{level}}

\newcommand{\resetcount}{\mathtt{resetcount}}
\newcommand{\resetcountmax}{R_{\max}}
\newcommand{\delay}{\mathtt{delaytimer}}
\newcommand{\delaymax}{D_{\max}}
\newcommand{\errorcountmax}{E_{\max} = \Theta(n)}

\newcommand{\errorcount}{\mathtt{errorcount}}

\newcommand{\depth}{H}
\newcommand{\synccount}{S_{\max}}
\newcommand{\node}{\mathtt{node}}
\newcommand{\edge}{\mathtt{edge}}
\newcommand{\edgetimer}{\mathtt{timer}}
\newcommand{\edgetimercount}{T_H}
\newcommand{\sync}{\mathtt{sync}}
\newcommand{\children}{\mathtt{children}}

\newcommand{\rankself}{\rank_i}

\newcommand{\field}{\mathtt{field}}

\newcommand{\leader}{\mathtt{leader}}

\renewcommand{\rank}{\mathtt{rank}}
\renewcommand{\rankself}{\rank}

\newcommand{\yes}{\mathsf{Yes}}
\newcommand{\no}{\mathsf{No}}

\newcommand{\settled}{\mathsf{Settled}}
\newcommand{\unsettled}{\mathsf{Unsettled}}
\newcommand{\resetting}{\mathsf{Resetting}}

\newcommand{\collecting}{\mathsf{Collecting}}

\newcommand{\inconsistent}{\mathsf{Inconsistent}}

\newcommand{\role}{\mathtt{role}}
\newcommand{\tree}{\mathtt{tree}}
\newcommand{\roleself}{\mathtt{role}}

%% file: title-authors.tex

\title{Time-Optimal Self-Stabilizing Leader Election in Population Protocols}
\date{}

\ifarticle
    \author[1]{Janna Burman}
    \author[2]{Ho-Lin Chen}
    \author[2]{Hsueh-Ping Chen}
    \author[3]{David Doty}
    \author[1]{Thomas Nowak}
    \author[3]{Eric Severson}
    \author[4]{Chuan Xu}
    \affil[1]{
        Universit\'e Paris-Saclay, 
        {\tt \{janna.burman,thomas.nowak\}@lri.fr}
    }
    \affil[2]{
        National Taiwan University, 
        {\tt \{holinchen,r07921034\}@ntu.edu.tw}
    }
    \affil[3]{
        University of California, Davis, 
        {\tt \{doty,eseverson\}@ucdavis.edu}
    }
    \affil[4]{
        Inria Sophia-Antipolis, 
        {\tt xuchuan898@gmail.com}
    }
\fi

%% file: Sections/abstract.tex
\begin{abstract}
    We consider the standard population protocol model, where (\emph{a priori}) indistinguishable and anonymous agents interact in pairs according to uniformly random scheduling.
    The \emph{self-stabilizing leader election} problem requires the protocol to converge on a single leader agent from \emph{any} possible initial configuration.
    We initiate the study of time complexity of population protocols solving this problem in its original setting: with probability 1, in a complete communication graph.
    The only previously known protocol 
    by Cai, Izumi, and Wada\ [Theor.\ Comput.\ Syst.~50] runs in expected parallel time $\Theta(n^2)$
    and has the optimal number of~$n$ states in a population of~$n$ agents.
    The existing protocol has the additional property that it becomes silent, i.e., the agents' states eventually stop changing.
    
    Observing that any silent protocol solving self-stabilizing leader election requires $\Omega(n)$ expected parallel time,
    we introduce a silent protocol that uses optimal $O(n)$ parallel time and states.
    Without any silence constraints, we show that it is possible to solve self-stabilizing leader election in asymptotically optimal expected parallel time of $O(\log n)$,
    but using at least exponential states 
    (a quasipolynomial number of bits).
    All of our protocols 
    (and also that of Cai et al.)
    work by solving the more difficult \emph{ranking} problem:
    assigning agents the ranks 
    $1,\ldots,n$.
\end{abstract}



    

%% file: Sections/introduction.tex
\section{Introduction}\label{sec:intro}

\emph{Population protocols} \cite{DBLP:journals/dc/AngluinADFP06} are a popular and well established model of distributed computing,
originally motivated by passively mobile sensor networks. 
However, it also models population dynamics from various areas such as 
trust and rumor propagation in social networks \cite{Diamadi2001}, 
game theory dynamics \cite{DBLP:journals/corr/abs-0906-3256}, 
chemical reactions \cite{Gillespie1977,DBLP:journals/nc/SoloveichikCWB08}, 
and gene regulatory networks \cite{BowerB04}. 
Population protocols are a special-case variant of Petri nets and vector addition systems \cite{DBLP:journals/acta/EsparzaGLM17}.

This model considers computational \emph{agents} with no ability to control their schedule of communication. 
They are \emph{a priori} anonymous, indistinguishable, and mobile: interacting in pairs asynchronously and unpredictably.
At each step a pair of agents to interact is chosen uniformly at random.
Each agent observes the other's state,
updating its own according to the transition function.
A \emph{configuration} describes the global system state: the state of each of the $n$ agents.
The sequence of visited configurations describes a particular \emph{execution} of the protocol.
The goal of the protocol is to reach a desired set of configurations with probability~1.

It is common in population protocols to measure space/memory complexity by counting the potential number of states each agent can have.\footnote{
    The base-2 logarithm of this quantity is the standard space complexity measure of the number of bits required to represent each state (e.g., polynomial states = logarithmic space).
}
The model originally used constant-state protocols, 
i.e., the state set is independent of the population size $n$~\cite{DBLP:journals/dc/AngluinADFP06}. 
Recent studies relax this assumption and allow the number of states to depend on $n$, 
adding computational power to the model~\cite{DBLP:conf/icalp/GuerraouiR09,doi:10.2200/S00328ED1V01Y201101DCT006,DBLP:journals/mst/BournezCR18},
improving time complexity~\cite{DBLP:conf/sirocco/Rabie17,DBLP:conf/soda/AlistarhAG18,DBLP:conf/spaa/GasieniecSU19},
or tolerating faults~\cite{cai2012prove,DBLP:conf/icalp/GuerraouiR09,DBLP:journals/tcs/LunaFIISV19}. In the current work, for tolerating any number of  transient faults (in the framework of self-stabilization), such relaxation is necessary~\cite{cai2012prove,DBLP:conf/sirocco/SudoS0KM20} (see details below and Theorem \ref{thm:n-state-lower-bound}).

\paragraph*{Leader election.}
In the \emph{leader election} problem, 
the protocol should reach a configuration $C$ with only one agent 
marked as a ``leader'', 
where all configurations reachable from $C$ also have a single leader. 
When this happens, the protocol's execution is said to have \emph{stabilized}.\footnote{
    Some protocols~\cite{DBLP:conf/soda/GasieniecS18, DBLP:conf/podc/KosowskiU18} stabilize with probability 1, 
    but converge (elect a unique leader)
    long before stabilizing
    (become unable to change the number of leaders).
    In our protocols these two events typically coincide.
}
The time complexity of a protocol is measured by \emph{parallel time}, the number of interactions until stabilization, 
divided by the number of agents $n$.\footnote{
    This captures the intuition that interactions happen in parallel,
    effectively defining the time scale so that each agent participates in $O(1)$ interactions per time unit on average.
}

Leader election is an important paradigm in the design of distributed algorithms useful to achieve a well coordinated and efficient behavior in the network.  
For example, in the context of population protocols, given a leader, protocols can become exponentially faster  \cite{DBLP:journals/dc/AngluinAE08a,DBLP:conf/icalp/BellevilleDS17} or compact (using less memory states) \cite{DBLP:conf/stacs/BlondinEJ18}.
Moreover, some problems, like fault-tolerant counting, naming 
and bipartition become feasible, assuming a leader \cite{Beauquier2007,DBLP:conf/wdag/BurmanBS19,
DBLP:conf/opodis/YasumiOYI17}.

Leader election protocols have been extensively studied in the original setting where all agents start in the same pre-determined state (a non-self-stabilizing case, and in complete interaction graphs, i.e. where any pair of agents can interact).
For example, it was shown that the problem cannot be solved in 
$o(n)$
(parallel) time if agents have only 
$O(1)$
states \cite{DBLP:conf/wdag/DotyS15},
an upper bound later improved to $\leq \frac{1}{2} \log \log n$ states~\cite{DBLP:conf/soda/AlistarhAEGR17}. 
To circumvent this impossibility result, 
subsequent studies assume a non-constant state space, though relatively small (e.g., $O(\log n)$ or $O(\log \log n)$). 
Leader election has recently been shown to be solvable with optimal $O(\log n)$ parallel time and $O(\log \log n)$ states~\cite{DBLP:conf/stoc/BerenbrinkGK20},
improving on recent work meeting this space bound in time $O(\log^2 n)$~\cite{DBLP:conf/soda/GasieniecS18} and $O(\log n \log \log n)$~\cite{DBLP:conf/spaa/GasieniecSU19}.
Another work presents a simpler $O(\log n)$-time, $O(\log n)$-state protocol~\cite{sudo2020logarithmic}.
It may appear obvious that any leader election protocol requires $\Omega(\log n)$ time, but this requires a nontrivial proof~\cite{sudo2020leader}.
There is also an $O(1)$-space and expected 
$O(\log^2 n)$-time
protocol, but 
with a positive error probability; 
and a slower 
$o(n)$-time 
(e.g., $\sqrt{n}$) protocol correct with probability 1~\cite{DBLP:conf/podc/KosowskiU18}.
Recent surveys \cite{DBLP:journals/sigact/AlistarhG18, DBLP:journals/eatcs/ElsasserR18} 
explain the relevant techniques. 

\paragraph*{Reliable leader election.}
The current paper studies leader election in the context of \emph{reliability}.
What if agents are prone to memory or communication errors?
What if errors cannot be directly detected, so agents cannot be re-initialized in response? 
As a motivating scenario one can imagine mobile sensor networks 
for mission critical and safety relevant applications 
{where rapid recovery from faults takes precedence over memory requirements.}
Imagine applications operating on relatively small sized networks, so that the sensors' memory storage dependent on $n$ is not necessarily an issue.
{Additionally, $n$ states are provably required to solve our problem~\cite{cai2012prove} (see Theorem \ref{thm:n-state-lower-bound}).}

We adopt the approach of self-stabilization~\cite{dijkstra,DBLP:conf/opodis/AngluinAFJ05}.
A protocol is called \emph{self-stabilizing} if it stabilizes with probability 1 from an \emph{arbitrary} configuration\footnote{
    %
    For a self-stabilizing protocol, it is equivalent to consider probability 1 and fixed probability $p > 0$ of correctness; See Section~\ref{sec:preliminaries}. 
}
(resulting from any number of transient faults).
Non-self-stabilizing (a.k.a., \emph{initialized}) leader election is easily solvable using only one bit of memory per agent by the single transition $(\ell,\ell) \to (\ell,f)$ from an initial configuration of all $\ell$'s: 
when two candidate leaders meet, one becomes a follower $f$.
However, this protocol fails 
(as do nearly all other published leader election protocols) 
in the self-stabilizing setting from an all-$f$ configuration.
Thus, any self-stabilizing leader election (SSLE) protocol must be able not only to reduce multiple potential leaders to one,
but also to create new leaders.
A particular challenge here is a careful verification of leader absence, to avoid creating excess leaders forever.

Because of this challenge, in any SSLE protocol, 
agents must know the \emph{exact} population size $n$, 
and the number of states must be at least $n$~\cite{cai2012prove} (Theorem \ref{thm:n-state-lower-bound} in the preliminaries section).
Despite the original assumption of constant space, population protocols with linear space 
(merely $O(\log n)$ bits of memory)
may be useful in practice, similarly to distributed algorithms in other models (message passing, radio networks, etc.).  
One may now imagine such memory-equipped devices communicating in a way as agents do in population protocols \cite{DBLP:conf/sensys/PolastreHC04, DBLP:conf/infocom/JohnsonSFFSRL06}. Think of a group of mobile devices (like sensors, drones or smart phones) operating in different types of rescue, military or other monitoring operations (of traffic, pollution, agriculture, wild-life, etc.). Such networks may be expected to operate in harsh inaccessible environments, while being highly reliable and efficient. This requires an efficient ``strong’’ fault-tolerance in form of automatic recovery provided by self-stabilization. 
Moreover, even if one considers only protocols with polylog$(n)$ states interesting, it remains an interesting fact that such protocols cannot solve SSLE.


Finally, self-stabilizing algorithms are easier to compose~\cite{DBLP:journals/dc/DolevIM93,dolev2000self}.
Composition is in general difficult for population protocols~\cite{severson2020composable, chalk2021composable}, since they lack a mechanism to detect when one computation has finished before beginning another.
However, a self-stabilizing protocol $S$ can be composed with a prior computation $P$, which may have set the states of $S$ in some unknown way before $P$ stabilized, 
c.f.~\cite[Section 4]{DBLP:conf/opodis/AngluinAFJ05}, \cite[Theorem 3.5]{amir2020message}.

\paragraph{Problem variants.}
To circumvent the necessary dependence on population size $n$,
previous work has considered relaxations of the original problem.
One approach, which requires agents only to know an upper bound on $n$, 
is to relax the requirement of self-stabilization: 
\emph{loose-stabilization} requires only that a unique leader persists for a long time after a stabilization, 
but not forever~\cite{DBLP:journals/tcs/SudoOKMDL20}. 
Other papers study leader election in more general and powerful models than population protocols,
which allow extra computational ability not subject to the limitations of the standard model. 
One such model assumes an external entity, called an \emph{oracle}, 
giving clues to agents about the existence of leaders~\cite{DBLP:conf/opodis/FischerJ06, DBLP:conf/opodis/BeauquierBB13}.
    Other generalized models include
\emph{mediated population protocols}~\cite{DBLP:journals/dc/MizoguchiOKY12},
allowing additional shared memory for every pair of agents,
and 
the $k$-interaction model~\cite{DBLP:conf/sss/XuYKY13},
where agents interact in groups of size $2$ to $k$.

While this paper considers only the complete graph (the most difficult case),
other work considers protocols that assume a particular non-complete graph topology.
In rings and regular graphs with constant degree, SSLE is feasible even with only a constant state space \cite{chen2019self,chen2020self,DBLP:conf/opodis/AngluinAFJ05,DBLP:conf/sss/YokotaSM20}. 
In another recent related work \cite{DBLP:conf/sirocco/SudoS0KM20}, the authors study the feasibility requirements of SSLE in arbitrary graphs, as well as the problem of \emph{ranking} that we also study (see below).
They show how to adapt protocols in \cite{DBLP:conf/opodis/BeauquierBB13, cai2012prove} 
into protocols 
for an arbitrary (and unknown) connected graph topology (without any oracles, but knowing $n$).

\subsection{Contribution}

\begin{table}
\ifarticle
    \small 
\fi
\caption{
    \small
    Overview of time and space (number of states) complexities of self-stabilizing leader election protocols (which all also solve ranking). 
    For the silent protocols, the silence time also obeys the stated upper bound. 
    Times are measured as parallel time until stabilization both in expectation and with high probability 
    (WHP is defined as probability $1 - O(1/n)$, but implies a guarantee for any $1-O(1/n^c)$, see Section~\ref{sec:preliminaries}).
    Entries marked with * are asymptotically optimal in their class (silent/non-silent);
    see Observation~\ref{obs:linear-time-lower-bound-silent}.
    The final two rows really describe the same protocol \logTimeProtocol; it is parameterized by the positive integer $H$; setting $H = \Theta(\log n)$ gives the time-optimal $O(\log n)$ time protocol.
}
\label{tab:overview}
\centering
\vspace{-0.1cm}
\begin{tabular}{| l | l | l | l | l |}
\hline
    {\bf protocol} 
& 
    {\bf expected time} 
&
    {\bf WHP time}
&
    {\bf states} 
&
    {\bf silent} 
\\ \hline
    \OriginalSSLE~\cite{cai2012prove} 
&
    \hphantom{*} $\Theta(n^2)$ 
&
    \hphantom{*} $\Theta(n^2)$
&
    * $n$ 
&
    yes
\\ \hline
    \silentlinearTimeStateProtocol\ (Sec.~\ref{sec:silent-linear-time-state})
&
    * $\Theta(n)$
&
    * $\Theta(n \log n)$
&
    * $O(n)$
&
    yes
\\ \hline
    \logTimeProtocol\ (Sec.~\ref{sec:log-time-protocol})
&
    * $\Theta(\log n)$ 
&
    * $\Theta(\log n)$
&
    \hphantom{*} $\exp(O(n^{\log n}\cdot \log n))$
&
    no
\\ \hline
    \logTimeProtocol\ (Sec.~\ref{sec:log-time-protocol})
&
    \hphantom{*} $\Theta(H \cdot n^{\frac{1}{H+1}})$ 
&
    \hphantom{*} $\Theta(\log n \cdot n^{\frac{1}{H+1}})$
&
    \hphantom{*} $\Theta(n^{\Theta(n^{H})}\log n)$
&
    no
\\ \hline

\end{tabular}
\vspace{-0.3cm}

\end{table}

We initiate the study of the limits
of time efficiency or the time/space trade-offs for SSLE in the standard population protocol model, in the complete interaction graph.
The most related protocol, of Cai, Izumi, and Wada~\cite{cai2012prove} (\OriginalSSLE, Protocol~\ref{algo:original}), given for complete graphs, uses exactly $n$ states and $\Theta(n^2)$ expected parallel time %
(see Theorem~\ref{thm:original-timebound})%
, exponentially slower than the $\mathrm{polylog}(n)$-time non-self-stabilizing existing solutions~\cite{DBLP:conf/stoc/BerenbrinkGK20, DBLP:conf/soda/GasieniecS18, DBLP:conf/spaa/GasieniecSU19, DBLP:conf/podc/KosowskiU18, sudo2020logarithmic}.
Our main results are two faster protocols, 
each making a different time/space tradeoff.

Our protocols, along with that of~\cite{cai2012prove},
are summarized in Table~\ref{tab:overview}.
These main results are later proven as Theorem~\ref{thm:silent-linear-time} and Theorem~\ref{thm:log-time}.
Both expected time and high-probability time are analyzed. 
Any \emph{silent} protocol 
(one guaranteed to reach a configuration where no agent subsequently changes states) 
must use $\Omega(n)$ parallel time in expectation
(Observation~\ref{obs:linear-time-lower-bound-silent}).
This lower bound has helped to guide our search for sublinear-time protocols, since it rules out ideas that, if they worked, would be silent.
Thus \silentlinearTimeStateProtocol\ is time- and space-optimal for the class of silent protocols.

\logTimeProtocol\ is actually
a family of sublinear time protocols that, depending on a parameter $H$ that can be set to an integer between 1 and $\Theta(\log n)$,
causes the algorithm's running time to lie somewhere in $O(\sqrt{n})$ and $O(\log n)$, while using more states the larger $H$ is;
setting $H = \Theta(\log n)$ gives the time-optimal $O(\log n)$ time protocol.
However, even with $H=1$, it requires exponential states.
It remains open to find a sublinear-time SSLE protocol that uses sub-exponential states.
We note that any protocol solving SSLE requires $\Omega(\log n)$ time:
from any configuration where all $n$ agents are leaders,
by a coupon collector argument, it takes $\Omega(\log n)$ time for $n-1$ of them to interact and become followers. 
(This argument uses the self-stabilizing assumption that ``all-leaders'' is a valid initial configuration; otherwise, for \emph{initialized} leader election, it requires considerably more care to prove an $\Omega(\log n)$ time lower bound~\cite{sudo2020logarithmic}.)

For some intuition behind the parameterized running times for \logTimeProtocol,
the protocol works by detecting ``name collisions'' between agents, communicated via paths of length $H+1$.
For example, $H=0$ corresponds to the simple linear-time algorithm that relies on two agents $s,a$ with the same name directly interacting, i.e., the path $s \to a$.
$H=1$ means that $s$ first interacts with a third agent $b$, who then interacts with $a$, i.e., the path $s \to b \to a$.
To analyze the time for this process to detect a name collision,
consider the following ``bounded epidemic'' protocol. 
The ``source'' agent $s$ that starts the epidemic is in state $0$, and all others are in state $\infty$, and they interact by $i, j \rightarrow i, i+1$ whenever $i < j$. The time $\tau_k$ is the first time some target agent $a$ has state $\leq k$. In other words, this agent has heard the epidemic via a path from the source of length at most $k$.
We have $\IE[\tau_1] = O(n)$, since $a$ must meet $s$ directly. An iterative process can then show $\IE[\tau_2] = O(\sqrt{n})$, and more generally $\IE[\tau_k] = O(kn^{1/k})$.
$\tau_n$ is the hitting time for the standard epidemic process, since the path from any agent to the source can be at most $n$. However, with high probability, the epidemic process will reach each agent via a path of length $O(\log n)$, so it follows that $\tau_{k} = O(\log n)$ if $k = \Omega(\log n)$, so setting $H = \Theta(\log n)$ will detect this name collisions in $O(\log n)$ time. Bounds on $\tau_k$ are given as Lemma~\ref{lem:bounded-epidemic-constant} and Lemma~\ref{lem:bounded-epidemic-log-n}.

\opt{inline}{\paragraph*{Ranking.}}
All protocols in the table 
solve a more difficult problem than leader election:
\emph{ranking} the agents by assigning them the IDs $1,\ldots,n$.
Ranking is helpful for SSLE because it gives a deterministic way to detect the absence of a state (such as the leader state). If any rank is absent, the pigeonhole principle ensures multiple agents have the same rank, reducing the task of absence detection to that of collision detection.

Collision detection is accomplished easily in $O(n)$ time by waiting for the colliding agents to meet, which is done by
\silentlinearTimeStateProtocol. Achieving stable collision detection in optimal $O(\log n)$ time is key to our fast protocol \logTimeProtocol. This collision detection problem is interesting in its own right, see Conclusion.

Ranking is similar to the 
\emph{naming} problem of assigning each agent a unique ``name'' (ID) \cite{DBLP:conf/sss/MichailCS13,DBLP:conf/wdag/BurmanBS19}, but is strictly stronger since each agent furthermore knows the order of its name relative to those of other agents.
Naming is related to leader election:
if each agent can determine whether its name is ``smallest'' in the population,
then the unique agent with the smallest name becomes the leader.
However, it may not be straightforward to determine whether some agent exists with a smaller name;
much of the logic in the faster ranking algorithm \logTimeProtocol\ is devoted to propagating the set of names of other agents while determining whether the adversary has planted ``ghost'' names in this set that do not actually belong to any agent.
On the other hand, 
any ranking algorithm automatically solves both the naming and leader election problems:
ranks are unique names,
and the agent with rank 1 can be assigned as the leader.
(Observation~\ref{obs:leader-election-without-ranking} shows that the converse does not hold.)

%% file: Sections/preliminaries.tex
\section{Preliminaries}\label{sec:preliminaries}

We 
{write}
$\IN = \{1,2,\ldots\}$ 
and
$\IN_0 = \IN \cup \{0\}$. 
The term $\ln k$ denotes the natural logarithm of~$k$.
$H_k = \sum_{i=1}^k\frac{1}{i}$ denotes the $k$th harmonic number, with $H_k \sim \ln k$, where $f(k) \sim g(k)$ denotes that $\lim\limits_{k \to \infty} \frac{f(k)}{g(k)} = 1$.
We omit floors or ceilings 
{(which are asymptotically negligible)} when writing $\ln n$ to describe a quantity that should be integer-valued.
Throughout this paper, by convention $n$ denotes the \emph{population size} $n$, the number of agents.
We say an event $E$ happens \emph{with high probability (WHP)} if 
$\IP[\neg E] = O(1/n)$.

If a self-stabilizing protocol stabilizes with high probability, then we can make this high probability bound $1-O(1/n^c)$ for any constant $c$. This is because in the low probability of an error, we can repeat the argument, using the current configuration as the initial configuration. Each of these potential repetitions gives a new ``epoch'', where the Markovian property of the model ensures the events of stabilizing in each epoch are independent. Thus
the protocol will stabilize after at most $c$ of these ``epochs'' with probability $1-O(1/n^c)$. By the same argument, if a self-stabilizing protocol can stabilize with any positive probability $p>0$, it will eventually stabilize with probability 1.

\paragraph*{Model.} 
We consider population protocols \cite{DBLP:journals/dc/AngluinADFP06} defined on a collection $\Agents$ of $n$ indistinguishable agents, also called a population. We assume a complete communication graph over $\Agents$, meaning that every pair of agents can interact. Each agent has a set~$\mathbf{S}$ of local states. At each discrete step of a protocol, a \emph{probabilistic scheduler} picks randomly an ordered pair of agents from $\Agents$ to interact.
During an interaction, the two agents mutually observe their states and update them according to
a probabilistic\footnote{Note that we allow randomness in the transitions for ease of presentation. All our protocols can be made deterministic by standard synthetic coin techniques without changing time or space bounds; see Section~\ref{sec:derandomization}.} transition function
$\Ttable : \mathbf{S}\times \mathbf{S} \to \Dist(\mathbf{S}\times \mathbf{S})$ where
$\Dist(X)$ denotes the set of probability distributions on~$X$.

Given a finite population $\Agents$ and state set $\mathbf{S}$, we define a \emph{configuration} $C$ as a mapping $C:\Agents\to \mathbf{S}$. Given a starting configuration $C_0$, we define the corresponding {\em execution\/} as a sequence $(C_t)_{t\geq 0}$ of random configurations where each $C_{t+1}$ is obtained from $C_t$ by applying $\Ttable$ on the states of a uniform random ordered pair of agents $(a,b)$, 
i.e., $C_{t+1}(a),C_{t+1}(b)=\Ttable(C_t(a),C_t(b))$ and $C_{t+1}(x)=C_t(x)$ for all $x\in \Agents\setminus \{a,b\}$.
We use the word \emph{time} to mean the number of interactions divided by~$n$ (the number of agents), a.k.a.\  \emph{parallel time}.

    \paragraph*{Pseudocode conventions.}
    We describe states of agents by several \emph{fields},
    using fixed-width font to refer to a field such as $\field$.
    As a convention, we denote by $a.\field(t)$,
    when used outside of pseudocode,
    the value of $\field$ in agent $a$ at the end of the $t$\textsuperscript{th} interaction,
    omitting ``$a.$'' and/or ``$(t)$'' when the agent and/or interaction is clear from context.
    Constant values are displayed in a sans serif front such as $\mathsf{Yes}$/$\mathsf{No}.$
    When two agents $a$ and $b$ interact,
    we describe the update of each of them using pseudocode, where we refer to $\field$ of agent $i\in\{a,b\}$ as  $i.\field$.
    
    In each interaction, one agent is randomly chosen by the scheduler to be the ``initiator'' and the other the ``responder''.
    Most interactions are symmetric, so we do not explicitly label the initiator and responder unless an asymmetric interaction is required.\footnote{It is also possible to make all transitions symmetric using standard ``synthetic coin'' techniques; see Section~\ref{sec:derandomization}.}
    
    A special type of field is called a \emph{role}, used in some of our protocols to optimize space usage and limit the types of states accessible to an adversarial initial condition.
    If an agent has several fields each from a certain set, 
    then that agent's potential set of states is the cross product of all the sets for each field,
    i.e., adding a field from a set of size $k$ multiplies the number of states by $k$.
    A role is used to \emph{partition} the state space:
    different roles correspond to different sets of fields,
    so switching roles amounts to deleting the fields from the previous role.
    Thus the total number of states is obtained by \emph{adding} the number of states in each role.

\paragraph*{Convergence and stabilization.}
Population protocols have some problem-dependent notion of ``correct'' configurations.
(For example, a configuration with a single leader is ``correct'' for leader election.)
A configuration $C$ is \emph{stably correct} if every configuration reachable from $C$ is correct.
An execution $\mathcal{E} = (C_0,C_1,\ldots)$ is picked at random according to the scheduler explained above. 
We say $\mathcal{E}$ \emph{converges} (respectively, \emph{stabilizes)} at interaction $i \in \IN$ if $C_{i-1}$ is not correct (resp., stably correct) and for all $j \geq i$, $C_j$ is correct (resp., stably correct).
The {\em (parallel) convergence/stabilization time\/} of a protocol is defined as the number of iterations to converge/stabilize, divided by $n$. 
Convergence can happen strictly before stabilization,
although a protocol with a bounded number of states converges from a configuration $C$ with probability $p \in [0,1]$ 
if and only if it stabilizes from $C$ with probability $p$.
For a computational task $T$ equipped with some definition of ``correct'',
we say that a protocol \emph{stably computes $T$ with probability $p$} if,
with probability $p$, it stabilizes (equivalently, converges).

\paragraph*{Leader election and ranking.}
The two tasks we study in this paper are self-stabilizing \emph{leader election} (SSLE) and \emph{ranking} (SSR).
For both, the self-stabilizing requirement states that from any configuration, 
a stably correct configuration must be reached with probability 1.
For leader election, each agent has a field $\leader$ with potential values $\{\yes,\no\}$,
 and a \emph{correct} configuration is defined
where exactly one agent $a$ has $a.\leader = \yes$.\footnote{
    We do not stipulate the stricter requirement that one agent stays the leader, rather than letting the $\leader=\yes$ bit swap among agents, but we claim these problems are equivalent due to the complete communication graph.
    A protocol solving SSLE can also ``immobilize'' the unique $\leader=\yes$ bit by replacing any transition $(x,y) \to (w,z)$, where $x.\leader = z.\leader = \yes$ and $y.\leader = w.\leader = \no$,
    with $(x,y) \to (z,w)$. 
}
For ranking, each agent has a field $\rank$ with potential values $\{1,\ldots,n\}$,
and a \emph{correct} configuration is defined as one where,
for each $r \in \{1,\ldots,n\}$,
exactly one agent $a$ has $a.\rank = r$.
As noted in Sec.~\ref{sec:intro}, any protocol solving SSR also solves SSLE by assigning $\leader$ to $\yes$ if and only if $\rank = 1$;
for brevity we omit the $\leader$ bit from our protocols and focus solely on the ranking problem.
\opt{inline,append}{Observation~\ref{obs:leader-election-without-ranking} shows that the converse does not hold.}

\paragraph*{SSLE Protocol from \cite{cai2012prove}.}
Protocol~\ref{algo:original} shows the original SSLE protocol from \cite{cai2012prove}. We display it here to introduce our pseudocode style and make it clear that this protocol is also solving ranking.\footnote{Their state set $\{0,\ldots,n-1\}$ from \cite{cai2012prove} is clearly equivalent to our formal definition of a $\rankself\in\{1,\ldots,n\}$, but simplifies the modular arithmetic.}

The convergence proofs in \cite{cai2012prove} did not consider our definition of parallel time via the uniform random scheduler. Thus we also include proofs that \OriginalSSLE\ stabilizes in $\Theta(n^2)$ time, in expectation and WHP (see Theorem~\ref{thm:original-timebound}). It is straightforward to argue an $\Omega(n^2)$ time lower bound from a configuration with  $2$ agents at $\rankself = 0$, $0$ agents at $\rankself = n-1$, and $1$ agent at every other $\rankself$.
This requires $n-1$ consecutive ``bottleneck'' transitions,
each moving an agent up by one rank starting at 0.
Each takes expected time $\Theta(n)$ since two specific agents (the two with the same rank) must interact directly.
Our arguments for a $O(n^2)$ time upper bound give a separate proof of correctness from that in \cite{cai2012prove}, reasoning about a barrier rank that is never crossed.

\begin{protocol}[H]
    \caption{\OriginalSSLE, for initiator $a$ interacting with responder $b$
    \\
    \textbf{Fields:} $\rankself \in \{ 0,\ldots,n-1\}$}
    \label{algo:original}
    \begin{algorithmic}[1]
        \If{$a.\rankself=b.\rankself$}
            \State{$b.\rankself\gets (b.\rankself+1) \bmod n$ }
        \EndIf
    \end{algorithmic}
\end{protocol}

Cai, Izumi, and Wada~\cite{cai2012prove} show that the state complexity of this protocol is optimal.
A protocol is \emph{strongly nonuniform} if, for any $n_1 < n_2$, a different set of transitions is used for populations of size $n_1$ and those of size $n_2$ (intuitively, the agents hardcode the exact value~$n$).

\begin{theorem}[\cite{cai2012prove}]
\label{thm:n-state-lower-bound}
    Any population protocol solving SSLE has $\geq n$ states and is strongly nonuniform.
\end{theorem}

It is worth seeing why any SSLE protocol must be strongly nonuniform.
Suppose the same transitions are used in population sizes $n_1 < n_2$.
By identifying in a single-leader population of size $n_2$ any subpopulation of size $n_1$ that does not contain the leader,
sufficiently many interactions strictly within the subpopulation must eventually produce a second leader.
Thus the full population cannot be stable. These conflicting requirements to both produce a new leader from a leaderless configuration, 
but also make sure the single-leader configuration is stable, 
is the key new challenge of leader election in the self-stabilizing setting.
Protocols solving SSLE circumvent this error by using knowledge of the exact population size $n$.

\begin{toappendix}



We analyze the time complexity of the \OriginalSSLE\ protocol.
It crucially relies on the fact that every (initial) configuration guarantees the existence of a ``barrier rank'' that is never exceeded by an interaction, disallowing indefinite cycles.
More formally,
denote by $m_i(C)$ the number of agents with rank~$i$ in configuration~$C$.
We will show that, starting from a configuration~$C_0$, there exists some~$k$ such that
\begin{equation} \label{eq:original:barrier:sums}
    \forall r\in \{0,\dots,n-1\}\colon\quad
    \sum_{d=0}^{r} m_{(k-d)\bmod n}(C) \leq r+1
\end{equation}
for all configurations~$C$ that are reachable from~$C_0$.
Then~$k$ is a barrier rank as it guarantees $m_k(C) \leq 1$ during the whole execution.

\begin{lemma}\label{lem:original:barrier:exists}
    For every configuration~$C$ of \OriginalSSLE,
    there exists some $k \in \{0,\dots,n-1\}$ such that~\eqref{eq:original:barrier:sums} holds in~$C$.
\end{lemma}

\begin{proof}
    Define~$S_i = \sum_{j=0}^i \big(m_j(C)-1\big)$ for $i\in\{0,\dots,n-1\}$.
    Note that $ S_{n-1} = 0$ since $\sum_{j=1}^{n-1} m_j(C) = n$.
    Choose $k\in\{0,\dots,n-1\}$ such that~$S_k$ is minimal.
    For $r\leq k$ we have
    \begin{equation*}
        \sum_{d=0}^{r} m_{(k-d)\bmod n}(C)
        =
        \sum_{j=k-r}^{k} m_{j}(C)
        =
        (r+1) + \sum_{j=k-r}^{k} \big(m_{j}(C)-1\big)
        =
        (r+1) + \left(S_{k} - S_{k-r+1}\right)
        \leq r+1
    \end{equation*}
    since $S_k \leq S_{k-r+1}$.
    For $r > k$ we have
    \begin{equation*}
        \sum_{d=0}^{r} m_{(k-d)\bmod n}(C)
        =
        \sum_{j=0}^{k} m_{j}(C)
        +
        \sum_{j=n-r+k}^{n-1} m_j(C)
        =
        (r+1) + \left(S_{k} + S_{n-1} - S_{n-r+k-1}\right)
        \leq r+1
    \end{equation*}
    since $S_{n-1} = 0$ and $S_k \leq S_{n-r+k-1}$.
\end{proof}

\begin{lemma}\label{lem:original:barrier:in:time}
    Let $k\in\{0,\dots,n-1\}$.
    If~\eqref{eq:original:barrier:sums} holds for~$k$ in configuration~$C$, then~\eqref{eq:original:barrier:sums} holds
    for~$k$ in all configurations reachable from~$C$.
\end{lemma}

\begin{proof}
    Without loss of generality we assume $k=n-1$ by cyclic permutation of the ranks.
    It suffices to show the lemma's statement for a direct successor configuration~$C'$ of~$C$.
    Let~$i$ be the rank of the initiator and~$j$ that of the responder in the interaction leading from~$C$ to~$C'$.
    If $i\neq j$, then  $m_\ell(C')=m_\ell(C)$ for all $\ell\in \{0,\dots,n-1\}$ and thus the statement follows from the hypothesis on~$C$.
    So assume $i=j$ in the rest of the proof.
    Then $i < n-1$ since $m_{n-1}(C)\leq 1$ using~\eqref{eq:original:barrier:sums} with $r=0$.
    Thus $m_i(C') = m_i(C)-1$ and $m_{i+1}(C') = m_{i+1}(C) + 1$.
    This means that all sums except for $r = n-i-2$ in~\eqref{eq:original:barrier:sums} remain constant when passing from~$C$ to~$C'$.
    
    To bound the sum for $r = n-i-2$,  we prove that we actually have $\sum_{d=0}^{r} m_{k-d}(C) \leq r$,
    which then implies $\sum_{d=0}^{r} m_{k-d}(C') = 1+\sum_{d=0}^{r} m_{k-d}(C) \leq r+1$ as required.
    In fact, if $\sum_{d=0}^{r} m_{k-d}(C) = r+1$, then
    \begin{equation*}
        m_i(C)
        =
        \sum_{d=0}^{r+1} m_{k-d}(C)
        -
        \sum_{d=0}^{r} m_{k-d}(C)
        \leq
        (r+2) - (r+1)
        =
        1
        \enspace,
    \end{equation*}
    which is a contradiction to the fact that there are two different agents with rank~$i$ in configuration~$C$ that can interact.
\end{proof}

\begin{theorem}
\label{thm:original-timebound}
    The \OriginalSSLE\ protocol solves self-stabilizing ranking. The silence time from the worst-case initial configuration is $\Theta(n^2)$ in expectation and with probability $1-\exp(-\Theta(n))$.
\end{theorem}

\begin{proof}
    We first prove the time lower bound.
    We let the protocol start in initial configuration~$C_0$ with $m_1(C_0) = 2$, $m_{n-1}(C_0) = 0$, and $m_i(C_0)=1$ for all $i\in\{1,\dots,n-2\}$.
    If we define $T_{-1},T_0,T_1,\dots,T_{n-1}$
    by $T_{-1}=0$ and
    $T_{s} = \min\{t\geq T_{s-1} \mid m_s(t) = 1\}$,
    then~$T_i$ is the first time that two agents with rank $i-1$ interact.
    As there is at most one rank with more than one agent, the difference $T_{i}-T_{i-1}$ is equal to the number of interactions until the two agents with rank~$i$ interact, which is a geometric random variable with probability of success $p = 1/\binom{n}{2}$.
    Then $\IE\,T_{n-1} = \sum_{i=0}^{n-1}\IE(T_{i}-T_{i-1})=(n-1)\binom{n}{2}=\Theta(n^3)$.
    As $T_{n-1}$ is the sum of independent geometric random variables, we will also use Theorem 3.1 from \cite{janson2018tail}, where $\mu = \IE\,T_{n-1}$ and $\lambda = 1/2$, to show the lower tail bound
    \[
    \IP[T_{n-1} \leq \lambda \mu] 
    \leq \exp\qty(-p\mu(\lambda - 1 - \ln \lambda))
    = \exp(-\Theta(n)).
    \]
    Thus the convergence time from $C_0$ is $\Omega(n^2)$ ($\Omega(n^3)$ interactions) in expectation and with probability $1-\exp(-\Theta(n))$. This concludes the proof of the time lower bound.

    We now turn to the time upper bound. Now let the initial configuration $C_0$ be arbitrary.
    By Lemmas~\ref{lem:original:barrier:exists} and~\ref{lem:original:barrier:in:time} we get the existence of some $k\in\{0,\dots,n-1\}$ such that~\eqref{eq:original:barrier:sums} holds for all $t\in\IN_0$.
    In particular $m_k(t) \leq 1$, which means that rank~$k$ is indeed a barrier for every execution starting in~$C_0$.
    Without loss of generality, by cyclic permutation of the ranks, we assume $k=n-1$.
    We inductively define $T_{-1},T_0,T_1,\dots,T_{n-1}$
    by $T_{-1}=0$ and
    $T_{s} = \min\{t\geq T_{s-1} \mid m_s(t) = 1\}$.
    Then $m_s(t)=1$ for all $t\geq T_s$.
    
    Now $T_s-T_{s-1}$ is the number of interactions for rank collisions to reduce the count $m_s$ to $1$. This is stochastically dominated by the convergence of the classic ``fratricide'' leader election ($L+L\to L+F$, starting from all $L$). Letting $F$ be the number of interactions for the fratricide process to converge to a single $L$, we have $F$ is the sum of independent (non-identical) geometric random variables, where $\IE\, F = \Theta(n^2)$ and the minimum parameter $p_* = 1/\binom{n}{2}$.
    We then have $T_{n-1}$ stochastically dominated by a sum $S = \sum_{i=1}^{n-1}F_i$ of $n-1$ independent copies of $F$ (which is a sum of $\Theta(n^2)$ geometric random variables).
    Then $\IE\, T_{n-1} \leq \IE\, S = (n-1)\IE\, F = O(n^3)$.
    Also, we can use Theorem 2.1 of \cite{janson2018tail} to show an upper tail bound (where $\mu = \IE\, S$ and $\lambda = 3/2$)
    \[
    \IP[S \geq \lambda \mu] 
    \leq \exp\qty(-p_*\mu(\lambda - 1 - \ln \lambda))
    = \exp(-\Theta(n)).
    \]
    Thus from any initial configuration $C_0$, the convergence time is $O(n^2)$ ($O(n^3)$ interactions) in expectation and with probability $1-\exp(-\Theta(n))$.
\end{proof}
\end{toappendix}

\begin{toappendix}
    \begin{observationrep}\label{obs:leader-election-without-ranking}
        There is a silent SSLE protocol whose states cannot be assigned ranks such that it also solves the SSR problem.
    \end{observationrep}
    
    \begin{proof}
        The following protocol solves silent SSLE for a population size $n=3$. (Note the construction from \cite{cai2012prove} in Protocol~\ref{algo:original} is strictly better protocol, the purpose of this construction is just to show an example solving silent SSLE without solving ranking).
        
        The state set is $S=\{l\}\cup F$, where $F=\{f_0,f_1,f_2,f_3,f_4\}$. There will be exactly 5 silent configurations of the three agents: 
        $\{l,f_0,f_1\}, \{l,f_1,f_2\}, \{l,f_2,f_3\}, \{l,f_3,f_4\}, \{l,f_4,f_0\}$.  
        (In other words, a leader $l$ and two distinct followers $f_i,f_j$ with $|i-j| \equiv 1 \mod 5$).
        
        This can be easily accomplished by adding transitions from $(s,s)$ (for all states $s\in S$) and from $(f_i,f_j)$ (for all $f_i,f_j\in F$ with $|i-j| \not\equiv 1 \mod 5$) to a uniform random pair of states $(a,b)\in S\times S$. It is easily observed that starting from any configuration of 3 agents, this protocol must stabilize to one of the $5$ silent configurations above, and thus solves SSLE.
        
        However, there is no way to consistently assign the ranks $1,2,3$ to the states in the silent configurations. If WLOG we denote $l$ to be rank 1, then we must assign ranks 2 or 3 to each state in $F$. But since $|F|$ is odd, every such assignment places two states $f_i,f_j$ in the same rank where $|i-j| \equiv 1 \mod 5$. Since $\{l,f_i,f_j\}$ is a silent configuration
        that is incorrectly ranked, we have a contradiction.
    \end{proof}
    
    Note that Observation~\ref{obs:leader-election-without-ranking} does not rule out a more sophisticated way to transform any SSLE protocol to a SSR protocol.
    It merely rules out the most simple approach: assigning ranks to the existing states, without otherwise changing the protocol.

\end{toappendix}

\paragraph*{Silent protocols.}

A configuration $C$ is \emph{silent} if no transition is applicable to it
(put another way, every pair of states present in $C$ has only a null transition that does not alter the configuration).
A self-stabilizing protocol is \emph{silent} if,
with probability 1,
it reaches a silent configuration from every configuration.
Since convergence time $\leq$ stabilization time $\leq$ silence time, the following bound applies to all three.

\begin{observation}\label{obs:linear-time-lower-bound-silent}
    Any silent SSLE protocol has $\Omega(n)$ expected convergence time
    and
    for any $\alpha > 0$,
    probability 
    $\geq \frac{1}{2} n^{- 3 \alpha}$ to require $\geq \alpha n \ln n$ 
    convergence time.
\end{observation}

For example, letting $\alpha=1/3$, 
with probability $\geq \frac{1}{2n}$ 
the protocol requires $\geq \frac{1}{3} n \ln n$ time.

\begin{proof}
    Let $C$ be a silent configuration with a single agent in a leader state $\ell$.
    Let $C'$ be the configuration obtained by picking an arbitrary non-leader agent in $C$ and setting its state also to $\ell$.
    Since $C$ is silent and the states in $C'$ are a subset of those in $C$,
    no state in $C'$ other than $\ell$ can interact nontrivially with $\ell$.
    So the two $\ell$'s in $C'$ must interact to reduce the count of $\ell$.
    The number of interactions for this to happen is geometric with $\IP[\text{success}] = 1/\binom{n}{2} = \frac{2}{n(n-1)} < 3/n^2$, 
    so expected time $\geq n/3$ 
    and for any $\alpha > 0$,
    at least $\alpha n^2 \ln n$ interactions ($\alpha n \ln n$ time) are required with probability at least
    \[
        \qty( 1 - 3/n^2 )^{\alpha n^2 \ln n} 
    \geq 
        \frac{1}{2} e^{- 3 \alpha \ln n} 
    = 
        \frac{1}{2} n^{- 3 \alpha}. \qedhere
    \]
\end{proof}

\paragraph*{Probabilistic tools.}
An important foundational process is the {\em two-way epidemic process\/} for efficiently propagating a piece of information from a single agent to the whole population.

We also consider a generalization, the \emph{roll call process}, where every agent simultaneously propagates a unique piece of information (its \emph{name}). 
We build upon bounds from \cite{mocquard2016analysis} to show this process is also efficient (only $1.5$ times slower than the original epidemic process). This process appears in 
two of our protocols, but also gives upper bounds on the time needed for any parallel information propagation, since after the roll call process completes, every agent has had a chance to ``hear from'' every other agent. 
The analysis of these two processes gives tight large deviation bounds with specific constants. While getting these precise constants was more than 
what is strictly
necessary for the proofs in this work, the analysis of these processes may be of independent interest. This roll call process has been independently analyzed in 
\cite{casteigts2020sharp, RandomExchangesHaigh, RandomExchangesSteele, RandomExchangesMoon}.
\opt{append}{The analysis appears in the Appendix.}



\begin{toappendix}
\subsection{Probabilistic Tools}
    In the {\em two-way epidemic process\/}, agents have a variable $\mathtt{infected}\in\{\true,\false\}$ updating as $a.\mathtt{infected},b.\mathtt{infected} \gets (a.\mathtt{infected}\ \vee\ b.\mathtt{infected})$.
    Mocquard, Sericola, Robert, and Anceaume~\cite{mocquard2016analysis} gave an in-depth analysis of the two-way epidemic process. 
    This analysis gives upper bounds for many processes in our protocols. 
    In any process where some field value is propagated this way,
    in addition to other transitions or initial conditions with more than one infected agent,
    which may speed up the propagation but cannot hinder it, 
    we denote that process a {\em superepidemic}.
    The number of interactions $X$ to spread to the whole population is clearly stochastically dominated by the two-way epidemic.\footnote{
        We note that this sort of process, which is stochastically dominated by a ``pure epidemic'', 
        is generally the sort of process studied in most population protocols papers that use the term \emph{epidemic}.
    }
    Consequently, we state the results below for normal epidemics, but use them to reason about superepidemics.
    
    The next lemma uses results of~\cite{mocquard2016analysis}
    to prove a simplified upper tail bound.
    
    \begin{lemma}[\cite{mocquard2016analysis}]
    \label{lem:epidemic:upper:tail}
        Starting from a population of size $n$ with a single infected agent, let $T_n$ be the number of interactions until $a.\mathtt{infected}=\true$ for all $a\in\Agents$. 
        Then 
        $\IE[T_n] = (n-1)H_{n-1} \sim n \ln n$, and for $n \geq 8$ and $\delta \geq 0$,
        \[
            \IP[ T_n > (1+\delta) \IE[ T_n ] ]
            \leq 
            2.5 \ln(n) \cdot n^{-2\delta}
        .
        \]
    \end{lemma}
    
    \begin{proof}
        From~\cite{mocquard2016analysis} we have $\IE[T_n]=(n-1)H_{n-1}\sim n\ln n$. 
        Also from~\cite{mocquard2016analysis}, for any $n\geq 3$ and $c\geq 1$, we have large deviation bound
        \begin{align*}
            \IP[T_n>c\IE[T_n]]
        &\leq 
            f(c,n) 
        =
            \qty(1+\frac{2c(n-1)H_{n-1}(n-2)^2}{n})\times\qty(1-\frac{2}{n})^{c(n-1)H_{n-1}-2}
        \\&\leq
            \qty(\frac{1}{(1-2/n)^2}+\frac{2c(n-1)H_{n-1}(n-2)^2}{n(1-2/n)^2})\qty(e^{-\frac{2}{n}})^{c(n-1)H_{n-1}}
        \\&=
            \qty(\frac{1}{(1-2/n)^2}-2cnH_{n-1}+2cn^2H_{n-1})\exp(-2c\frac{n-1}{n}
            H_{n-1})
        \\&\leq
            \qty(\frac{1}{(1-2/3)^2}-2\cdot 3 \cdot H_{3-1} + 2cn^2H_{n-1})\exp(-2c\frac{n-1}{n}
            H_{n-1})
        \\&=
            \qty(0 + 2 c n^2 H_{n-1}) \exp(-2c\frac{n-1}{n} H_{n-1}).
        \end{align*}
        Now observe that $\frac{n-1}{n}H_{n-1}>\ln n+0.189$ for all $n\geq 8$. Then
        \begin{align*}
            \IP[T_n>c\IE[T_n]]
        &\leq 
            2cn^2H_{n-1}e^{-2c(\ln n+0.189)}
        \\&= 
            2H_{n-1}ce^{-0.378c}n^{2-2c}.
        \end{align*}
        Now we observe that $H_{n-1}<1.25\ln n$ for all $n\geq 8$ and $ce^{-0.378c}<1$ for all $c\geq 1$. These inequalities give
        \begin{align*}
            \IP[T_n>c\IE[T_n]]
        &\leq 
            2\cdot 1.25\ln n \cdot n^{2-2c} = 2.5\ln n \cdot n^{-2\delta}
        \end{align*}
        taking $c=1+\delta$.
    \end{proof}



    \begin{corollaryrep}
    \label{cor:epidemic:upper:tail}
    Define $T_n$ as in Lemma~\ref{lem:epidemic:upper:tail}. 
    Then $\IE[T_n] < 1.2n\ln n$ and $\IP[T_n > 3n \ln n] < \frac{1}{n^2}$.
    \end{corollaryrep}
    
    \begin{proof}
        Observe that $\IE[T_n]=(n-1)H_{n-1}<1.2n\ln n$ for all $n\geq 2$.
        Then 
        $\IE[T_n] < 1.2n\ln n$. 
        Also $3n\ln n > 2.5\IE[T_n]$, so by the upper tail bound of Lemma~\ref{lem:epidemic:upper:tail}, we have 
        \[
            \IP[T_n > 3n \ln n] \leq \IP[T_n > (1+1.5)\IE[T_n]] \leq 2.5\ln n\cdot n^{-3} \leq n^{-2}
        \]
        since $n > 2.5\ln n$ for all $n \geq 2$.
    \end{proof}
    
    We now consider a variation called the \emph{roll call process}, 
    where every agent starts with a $\Met$ containing a single entry: their unique ID. 
    The agents update with $a.\Met\gets(a.\Met\cup b.\Met)$. 
    Let $R_n$ be the number of interactions to reach the terminal configuration where $\mathtt{a.set}$ contains all $n$ IDs for every $a\in\Agents$.
    
    Again, we will consider processes that are stochastically dominated by $R_n$. We can view the roll call process as the spreading of $n$ epidemics in parallel. Note that the roll call process as described takes exponential states, but it also gives an upper bound for any constant number of epidemics spreading in parallel. We find that asymptotically $R_n$ is 1.5 times larger than $T_n$.
    This result about the expected value was shown independently in \cite{casteigts2020sharp, RandomExchangesMoon,RandomExchangesSteele,RandomExchangesHaigh}.
    The results we give here use a different technique which also gives our required large deviation bounds on the time for the roll call process.
 
    \begin{lemmarep}
    \label{lem:parallel:epidemics}
        Let $R_n$ be the number of interactions for the roll call process to complete. Then $\IE[R_n] \sim 1.5n\ln n$. Also $\IP[R_n > 3n\ln n] < \frac{1}{n}$.
    \end{lemmarep}
    
    \begin{proof}
        Notice that in the roll call process, each individual ID spreads as a two-way epidemic. Thus we have $n$ epidemic processes happening in parallel; however they are not independent.
    
        We start by observing a lower bound for $\IE[R_n]$.
        
        First it is necessary for every agent to have an interaction. Let $\IE_1$ be the expected number of interactions for every agent to interact. This is a coupon collector process where we select two agents (coupons) at each step. It follows from a standard coupon collector analysis that $\IE_1\sim\frac{1}{2}n\ln n$.
        
        It is then necessary for the last agent to be picked to spread their ID to the whole population. Let $\IE_2$ be the expected number of interactions for this ID to spread to the whole population,
        starting from this agent's first interaction.
        This is a standard epidemic process (starting with two infected agents, which is an asymptotically negligible difference), 
        so by Lemma~\ref{lem:epidemic:upper:tail} $\IE_2\sim n\ln n$ interactions. 
        Then $\IE[R_n] \geq \IE_1 + \IE_2 \sim 1.5n \ln n$.
        (Note that the entire process may still be incomplete by this point.)
        
        Now we can get an upper tail bound on $R_n$ by considering it as the maximum of $n$ (non independent) epidemic processes. 
        Taking the union bound with Lemma~\ref{lem:epidemic:upper:tail} gives
        \[
            \IP[R_n > (1 + \delta)\IE[T_n]]\leq n\cdot 2.5\ln n \cdot n^{-2\delta}
        \]
        and then taking $\delta = \frac{1}{2} + u$ for $u>0$ we have
        \[
            \IP\qty[R_n > \qty(1.5 + u)\IE[T_n]]\leq 2.5\ln n \cdot n^{-2u}
        \]
        Now since $R_n\geq 0$ we can compute $\IE[R_n]$ as
        \begin{align*}
            \IE[R_n]
        &=
            \int_{t=0}^\infty \IP[R_n > t]\ dt
        \\&\leq
            \int_{t=0}^{1.5\IE[T_n]} 1\ dt + \int_{1.5\IE[T_n]}^\infty \IP[R_n > t]\ dt
        \\&=
            1.5\IE[T_n] + \frac{1}{\IE[T_n]}\int_{0}^\infty \IP\qty[R_n > \qty(1.5 + u)\IE[T_n]]\ du
        \\&\leq
            1.5\IE[T_n] + \frac{1}{\IE[T_n]}\int_{0}^\infty 2.5\ln n \cdot n^{-2u}\ du
        \\&=
            1.5\IE[T_n] + \frac{2.5}{\IE[T_n]}\cdot-\frac{1}{2}n^{-2u}\Big|_0^{\infty}
        \\&=
            1.5\IE[T_n] + \frac{1.25}{\IE[T_n]} \sim 1.5n \ln n
        \end{align*}
        Thus we have $\IE[R_n] \sim 1.5n\ln n$.
        
        The observation that $\IP[R_n > 3n \ln n] < \frac{1}{n}$ then follows immediately from the same union bound and Corollary~\ref{cor:epidemic:upper:tail}.
    \end{proof}

    We next consider another variation called the \emph{bounded epidemic process}. Here some source agent $s$ has $s.\level = 0$, and when agents $a, b$ interact, $a$ updates as $a.\level \gets \min(a.\level, b.\level + 1)$ (and symmetrically for $b$). Let $a$ be a fixed target agent. We define the time $\tau_k$ to be the first time that $a.\level \leq k$.
    Intuitively, $\tau_k$ is the time at which $a$ hears information from the source $s$ through a path of length at most $k$.
    For example, $\tau_1$ is the time until $a$ and $s$ interact directly,
    and $\tau_2$ is the time until $a$ interacts with some agent who has already interacted with $s$.
    
    \begin{lemma}
    \label{lem:bounded-epidemic-constant}
    For any constant $k = O(1)$, $\IE[\tau_k] \leq kn^{1/k}$ and $\tau_k \leq n^{1/k}\cdot \frac{c}{2}\ln n$ with high probability $1 - 1/n^c$.
    \end{lemma}
    
    \begin{proof}
        For each $i = 0, 1,\ldots, k-1$, define $L_i$ to be the first time when the count of agents with $\level \leq i$ is at least $n^{i/k}$. Note $L_0 = 0$, since at time $0$ (and all future times) we have a single source agent at $\level = 0$.
        We will now show inductively that $\IE[L_i] \leq k n^{1/k}$ and $L_i = O(k n^{1/k})$.
        
        After time $L_i$, we have at least $n^{i/k}$ agents with $\level \leq i$. We next need to wait until time $L_{i+1}$, when at least $n^{(i+1)/k}$ agents have $\level \leq i+1$. Let $x(t)$ be the number of agents with $\level \leq i+1$ at time $t$, so $x(L_i) \geq n^{(i+1)/k}$. Then the probability $x$ increases in the next interaction is at least 
        \[
            p_x = \frac{\#(\level \leq i) \cdot \#(\level > i+1)}{\binom{n}{2}} \geq \frac{2 n^{i/k} (n-x)}{n^2}
            \geq \frac{n^{i/k}}{n},
        \]
        since $x \leq n^{i+1}/k \leq n^{(k-1)/k} < \frac{n}{2}$.
        Now the number of interactions $T = n (L_{i+1} - L_i)$ until time $L_{i+1}$ is bounded by a sum of independent geometric random variables with probability $p_x$, as $x$ ranges from $n^{i/k}$ to $n^{(i+1)/k}$. This gives
        \[
            \IE[T] \leq \sum_{x = n^{i/k}}^{n^{(i+1)/k}}\frac{n}{n^{i/k}}
            \leq 
            n^{(i+1)/k} \cdot \frac{n}{n^{i/k}} = n^{1+1/k}.
        \]
        Moving from interactions to parallel time, we get $\IE[L_{i+1} - L_i] \leq n^{1/k}$, so $\IE[L_{i+1}] \leq (i+1)n^{1/k}$ as desired.
        
        We can also get a high probability bound on $T$, using the Chernoff bound variant for independent geometric random variables from \cite{janson2018tail}, which will show that $L_i = O(kn^{1/k})$ with high probability. This will end up being negligible compared to the $O(\log n \cdot n^{1/k})$ high probability bound we need for the last step, so we omit the details here.
        
        After time $L_{k-1}$, we have a count $n^{(k-1)/k}$ agents with $\level \leq k-1$, and now must wait for the target agent $a$ to meet one of these agents, which will ensure $a.\level \leq k$. The number $T$ of required interactions is a single geometric random variable with probability $p \geq \frac{n^{(k-1)/k} \cdot 1}{\binom{n}{2}} \geq \frac{2}{n^{1+1/k}}$. Thus $\IE[T] \leq \frac{1}{2}n^{1+1/k}$, giving less than $n^{1/k}$ time, so $\IE[\tau_k] \leq kn^{1/k}$. For the high probability bound, 
        \[
        \IP[T \geq \frac{c}{2} \ln n n^{1+1/k}] \leq (1-p)^{c \ln n \cdot n^{1+1/k}} 
        \leq \exp(-\frac{2}{n^{1+1/k}} \cdot \frac{c}{2} \ln n \cdot n^{1+1/k})
        = n^{-c}.
        \]
        Since this bound is asymptotically larger than the high probability bounds for the inductive argument on earlier levels, we have $\tau_k \leq \frac{c}{2}n^{1/k}\ln n$ with high probability $1-1/n^c$.
    \end{proof}
    
    The following lemma is similar to Lemma~\ref{lem:bounded-epidemic-constant} but for $k = \Theta(\log n)$ rather than $k=O(1)$.
    
    \begin{lemma}
    \label{lem:bounded-epidemic-log-n}
    For $k = 3\log_2 n$, $\tau_k \leq 3\ln n$, in expectation and with probability $1-O(1/n^2)$.
    \end{lemma}

    \begin{proof}
        Consider the standard epidemic process starting at the source $s$. By Corollary~\ref{cor:epidemic:upper:tail}, it takes at most $3\ln n$ time to complete with high probability $1-O(1/n^2)$. We consider this epidemic process as generating a random tree in the population as follows. The source $s$ is the root, labeled 1. Then each other agent is a vertex whose parent is the agent who infected them. Furthermore, we can label each vertex by the order in which that agent was infected. This process will exactly create a uniform random recursive tree, as discussed in \cite{Drmota09}. By Theorem 3 in \cite{drmota2009height}, this tree has expected height $\IE[H] = e \ln n$, and there is a high probability bound $\IP[|H-E[H]| \geq \eta] \leq e^{-c\eta}$ for some constant $c$, which implies that $H = O(\log n)$ with high probability $1-1/n^{a}$ for any desired constant $a$.
        
        We then set $k$ to be this bound on the height $H$. It follows that when agent $a$ gets infected in the full epidemic process, $a$ also has $\level \leq k$, since the whole tree has height $\leq k$.
        
        It is also possible to analyze the depth of agent $a$ in this random tree more directly. In the worst case, $a$ has label $n$ (it is the last to be infected). Then for a vertex with label $i$, their parent's label is uniform in the range $\{1,\ldots, i-1\}$, since at the time this agent gets infected, who infected them has uniform probability over all currently infected agents. Thus we get a recursive process, where $T_0 \leq n$, and $T_i = \text{Uniform}(\{1,\ldots,T_{i-1} - 1\})$, giving a sequence of random variables. We have $\IE[T_i | T_{i-1}] = \frac{1 + (T_{i-1} - 1)}{2} = \frac{T_{i-1}}{2}$ when $T_{i-1} > 1$, and the depth of $a$ is $d = \min\{i:T_i = 1\}$.
        
        To get a high probability bound on $d$, we will formally define $T_i = \frac{1}{2}T_{i-1}$ for all $i > d$. This way, the recurrence $\IE[T_i | T_{i-1}] = \frac{T_{i-1}}{2}$ holds for all $i$. Unwrapping the recurrence, we get $\IE[T_i] = \frac{1}{2^i}T_0 \leq \frac{n}{2^i}$. Then setting $i = 3\log_2 n$ gives $\IE[T_{3\log_2 n}] \leq \frac{n}{n^3} = \frac{1}{n^2}$, and by Markov's Inequality, $\IP[T_{3\log_2 n} \geq 1] \leq \frac{1}{n^2}$. Thus choosing $k = 3\log_2 n$ ensures $a$ has depth $\leq k$ in the epidemic tree with probability $1-1/n^2$. Then using the union bound with the event the epidemic takes longer than $3\ln n$ time, we have $\tau_k \leq 3\ln n$ with probability $1-O(1/n^2)$.
    \end{proof}

\end{toappendix}

%% file: Sections/resetting-subprotocol.tex
\section{Resetting subprotocol}
\label{sec:resetting-subprotocol}

\propagateresetprotocol\ (Protocol \ref{algo:propag-reset}) is used as a subroutine in both of our protocols \silentlinearTimeStateProtocol\ (Sec.~\ref{sec:silent-linear-time-state}) and \logTimeProtocol\ (Sec.~\ref{sec:log-time-protocol}).
Intuitively, it provides a way for agents 
(upon detecting an error that indicates the starting configuration was ``illegal'' in some way) 
to ``reset'' quickly, after which they may be analyzed as though they began from the reset state.
For that, the protocol $\resetprotocol$ 
has to be defined for use by $\propagateresetprotocol$.
We assume that $\resetprotocol$ changes the $\roleself$ variable to something different from $\resetting$.
Crucially, after the reset,
agents have no information about whether a reset has happened and do not attempt any synchronization to ensure they only reset once,
lest the adversary simply sets every agent to believe it has already reset, preventing the necessary reset from ever occurring.\footnote{
    This is unlike in standard population protocol techniques in which ``phase information'' is carried in agents indicating whether they are encountering an agent ``before'' or ``after'' a new phase starts.
}

We now define some terms used in the analysis of
$\propagateresetprotocol$, and their intuition:

If $a.\roleself \neq \resetting$, then we say 
$a$ is \emph{computing} (it is executing the outside protocol).
Otherwise, for $a.\roleself=\resetting$, we use three terms. If $a.\resetcount = \resetcountmax$, we say $a$ is \emph{triggered} (it has just detected an error and initiated this global reset).
If $a.\resetcount > 0$ we say $a$ is \emph{propagating} (intuitively this property of positivity spreads by epidemic to restart the whole population; we also consider triggered agents to be propagating).
If $a.\resetcount = 0$, we say $a$ is \emph{dormant} (it is waiting for a delay to allow the entire population to become dormant before they start waking up, this prevents an agent from waking up multiple times during one reset).

Likewise, we will refer to a configuration as \emph{fully / partially} propagating (resp.\ dormant, computing, triggered) if all / some agents are propagating (resp.\ dormant, computing, triggered).

A configuration $C$ is \emph{awakening} if it is the first partially computing configuration reachable from a fully dormant configuration.
Protocols that use \propagateresetprotocol\ will start their analysis by reasoning about an awakening configuration, which formalizes the idea of having gone through a ``clean reset''.
In an awakening configuration, all agents are dormant except one agent who has just executed \resetprotocol. Computing agents will awaken dormant agents by epidemic, so within $O(\log n)$ time, all agents will have executed \resetprotocol\ once and then be back to executing the main algorithm.

\begin{protocol}[H]
\caption{\propagateresetprotocol($a$,$b$), for $\resetting$ agent $a$ interacting with agent $b$.
\\
\textbf{Fields:} If $\role = \resetting$, $\resetcount\in\{0,1,\ldots,\resetcountmax\}$ and when $\resetcount=0$ an additional field $\delay\in\{0,1,\ldots,\delaymax\}$.
}
\label{algo:propag-reset}
\begin{algorithmic}[1]
    \If{$a.\resetcount>0$ and $b.\role \neq \resetting$} 
        \quad\Comment{bring $b$ into $\resetting$ role}
            \State{$b \gets \resetting,
            \quad 
            b.\resetcount \gets 0,
            \quad
            b.\delay \gets \delaymax$}
            \label{algo:reset:other-assigned-to-resetting}
        \EndIf
    \If{$b.\role = \resetting$}
        \quad\Comment{change $\resetcount$}
        \State{$a.\resetcount,b.\resetcount \gets \max(a.\resetcount-1, b.\resetcount-1, 0)$}
        \label{algo:reset:resetcount:update}
    \EndIf
    
    \For{$i\in\{a,b\}$ with $i.\role = \resetting$ and $i.\resetcount = 0$}
        \quad\Comment{dormant agents}
        \If{$i.\resetcount$ just became 0}
            \quad\Comment{initialize $\delay$}
            \State{$i.\delay \gets \delaymax$}
        \Else
            \State{$i.\delay \gets i.\delay - 1$ }\label{algo:reset:delay:update}
        \EndIf
        \If{$i.\delay=0$ or $b.\role\neq\resetting$}
        \Comment{awaken by epidemic}
        \label{algo:if-test-for-executing-reset}
            \State execute \resetprotocol($i$)
            \qquad \Comment{\resetprotocol\ subroutine provided by main protocol}
            \label{algo:execute-reset}
        \EndIf
    \EndFor

\end{algorithmic}
\end{protocol}

We require $\resetcountmax=\Omega(\log n)$, and for our protocol will choose the concrete value $\resetcountmax = 60\ln n$. We also require $\delaymax=\Omega(\resetcountmax)$. For our $O(\log n)$ time protocol \logTimeProtocol, we have $\delaymax = \Theta(\log n)$. In $\silentlinearTimeStateProtocol$, we set $\delaymax = \Theta(n)$, to give enough time for the dormant agents to do a slow leader election so they finish reset with a unique leader.

\propagateresetprotocol\ begins by some agent becoming triggered ($\resetcount=\resetcountmax$). 
Although introduced for a different purpose, \propagateresetprotocol\ is essentially equivalent to a subprotocol used in~\cite{DBLP:conf/dna/AlistarhDKSU17}, so we adopt their time analysis to prove it completes in $O(\log n)$ time.
Briefly, from a partially triggered configuration, the propagating condition ($\resetcount > 0$) spreads by epidemic (in $O(\log n)$ time) (Lemma~\ref{lem:triggered-to-fully-propagating}). 
Once the configuration is fully propagating, it becomes fully dormant in $O(\log n)$ time (Lemma~\ref{lem:fully-propagating-to-fully-dormant}).
From the fully dormant configuration, we reach an awakening configuration within $O(\log n)$ time when the first agent executes $\resetprotocol$ (Theorem~\ref{thm:reset-time}). Then the instruction to execute $\resetprotocol$ spreads by epidemic (in $O(\log n)$ time). \opt{append}{The proofs of \propagateresetprotocol are given in Section~\ref{sec:propagate-reset-proofs}}

\begin{toappendix}

\subsection{Proofs for {\mdseries \propagateresetprotocol}}
\label{sec:propagate-reset-proofs}

We first observe (by lines~\ref{algo:reset:other-assigned-to-resetting} and \ref{algo:reset:resetcount:update} of \propagateresetprotocol) that we can analyze the $\resetcount$ field a using the definition from \cite{DBLP:journals/tcs/SudoOKMDL20} of a \emph{propagating variable} (that updates as $a,b\gets \max(a-1,b-1,0)$):

\begin{observation}
\label{obs:resetcount-is-propagating-variable}
    If we define the $\resetcount$ field for all agents by letting $a.\resetcount = 0$ for any computing agent $a$ ($a.\roleself \neq \resetting$), then in any interaction between $a,b\in\Agents$, their $\resetcount$ fields both become $\max(a.\resetcount-1,b.\resetcount-1,0)$.
\end{observation}

    \begin{lemmarep}\label{lem:triggered-to-fully-propagating}
        Using the specific value $\resetcountmax = 60\ln n$, starting from a partially triggered configuration, we reach a fully propagating configuration after at most $4 \ln n$ time with probability at least $1-O(1/n)$.
    \end{lemmarep}
    \begin{proof}
    Noting that $\resetcount$ is a propagating variable~\cite{DBLP:journals/tcs/SudoOKMDL20}, we can
    use the same proof as \cite[Corollary~8]{DBLP:journals/tcs/SudoOKMDL20}.
    The constant $\resetcountmax = 60\ln n$ matches the constant used in their result.
    \end{proof}
    
    Although we set $\resetcountmax = 60 \ln n$, 
    the following lemma holds for arbitrary positive $\resetcountmax$.
    
    \begin{lemmarep}\label{lem:fully-propagating-to-fully-dormant}
        Let $\resetcountmax \in \mathbb{N}^+$
        and
        $\delaymax = \Omega(\log n + \resetcountmax)$.
        Starting from a fully propagating configuration, we reach a fully dormant configuration after $O(\log n + \resetcountmax)$ time with high probability $1-O(1/n)$.
    \end{lemmarep}
    
    An equivalent process was analyzed in \cite{DBLP:conf/dna/AlistarhDKSU17}, and this proof follows Lemma 1 from \cite{DBLP:conf/dna/AlistarhDKSU17}.
    
    \begin{proof}
    We first assume no agents are computing, so no further agents will become triggered, and the $\resetcount$ field will only change as noted in Observation~\ref{obs:resetcount-is-propagating-variable}. We define a local potential $\Phi_t(a)=3^{a.\resetcount}$ for agent $a$ in $C_t$, except $\Phi_t(a)=0$ if $a.\resetcount=0$. We then define a global potential $\Phi_t=\sum_{a\in\Agents}\Phi_t(a)$.
    We assume as a worst case that every agent starts with $\resetcount=\resetcountmax$, so $\Phi_0\leq n3^{\resetcountmax}$.  The goal will now be to show this global potential drops quickly to $0$ (corresponding to a fully dormant configuration). We will show that $\Phi_t=0$ with high probability $1-O(1/n)$, where 
    $t = O(n(\log n + \resetcountmax))$.
    
    Now from Observation~\ref{obs:resetcount-is-propagating-variable}, if agents $a$ and $b$ interact in the $t^\text{th}$ interaction, we can further observe that
    $$\Phi_{t+1}(a)+\Phi_{t+1}(b)\leq 2/3\cdot(\Phi_t(a)+\Phi_t(b))$$
    (with equality when $\Phi_t(a)=0$ or $\Phi_t(b)=0$). We will thus have the change in global potential
    $$\Phi_{t+1}-\Phi_t\leq -1/3\cdot(\Phi_t(a)+\Phi_t(b))$$
    Then conditioning on the configuration $C_t$, we can compute
    \begin{align*}
    \IE[\Phi_{t+1}-\Phi_t|C_t]
    &\leq\sum_{\{a,b\}\in\Agents}\IP[a,b\text{ interact in $t^\text{th}$ interaction}]\cdot(-1/3)\cdot(\Phi_t(a)+\Phi_t(b)) \\
    & = \frac{1}{\binom{n}{2}}\cdot(-1/3)\sum_{a\in\Agents}(n-1)\Phi_t(a)
    = -\frac{2}{3n}\Phi_t
    \end{align*}
    and thus $\IE[\Phi_{t+1}|C_t]\leq(1-\frac{2}{3n})\Phi_t$, so $\IE[\Phi_{t+1}]=\IE[\IE[\Phi_{t+1}|C_t]]\leq(1-\frac{2}{3n})\IE[\Phi_t]$. Then by induction we have
    \[
    \IE[\Phi_t]
    \leq 
    \bigg(1-\frac{2}{3n}\bigg)^t\Phi_0
    \leq 
    \exp\bigg(-\frac{2}{3n}t\bigg) n3^{\resetcountmax}
    \]
    Thus we have $\IE[\Phi_t]\leq\frac{1}{n}$ for $t=\frac{3}{2}n\ln(n^2\cdot3^{\resetcountmax})=\frac{3}{2}n(2\ln n+ \resetcountmax\cdot\ln 3) = O(n(\log n + \resetcountmax))$, so by Markov's inequality we have $\IP[\Phi_t=0]\geq 1-\frac{1}{n}$.
    
    We assumed that no agents are computing during these $t$ interactions. For the first agent to become computing, $\delay$ must hit $0$ starting from $\delaymax$ (after that agent has become dormant). Choosing $\delaymax = \Omega(\log n + \resetcountmax)$, then no agent will have more than $\delaymax$ interactions during these $t$ interactions, with high probability $1-O(1/n)$ by standard Chernoff bounds. Thus all agents will become dormant before the first agent becomes computing with high probability.
    \end{proof}
    
    We now combine the previous lemmas to describe the behavior of $\propagateresetprotocol$ when initialized by a triggered agent. Recall that an awakening configuration is the first partially computing configuration reached from a fully dormant configuration, so the first agent has set $\delay = 0$ and executed $\resetprotocol$.
    
    \begin{theoremrep}
        \label{thm:reset-time}
        Let $\resetcountmax = 60 \ln n$ 
        and 
        $\delaymax = \Omega(\log n + \resetcountmax)$.
        Starting from a partially-triggered configuration, we reach an awakening configuration in at most $O(\delaymax)$ parallel time with probability at least $1-O(1/n)$.
    \end{theoremrep}
    
    \begin{proof}
        By Lemmas~\ref{lem:triggered-to-fully-propagating} and
        \ref{lem:fully-propagating-to-fully-dormant},
        we reach a fully dormant configuration after $O(n(\log n + \resetcountmax))$ interactions.
        
        Next, after an additional $n\frac{\delaymax}{2}$ interactions (which each include $2$ of the $n$ agents), by Pigeonhole Principle some agent must have participated in $\delaymax$ interactions. This agent had $\delay$ at most $\delaymax$ in the fully dormant configuration, so by lines~\ref{algo:reset:delay:update}-\ref{algo:execute-reset} of $\propagateresetprotocol$ that agent has executed $\resetprotocol$. When the first such agent executes $\resetprotocol$, we reach an awakening configuration.
    \end{proof}
    
    Theorem~\ref{thm:reset-time} describes the intended behavior of $\propagateresetprotocol$. This will let protocols start their analysis from an awakening state, and show that given this ``clean reset'' the protocol will then stabilize. The remaining step in the analysis will be to show from any configuration, we quickly become triggered and enter $\propagateresetprotocol$ (or simply stabilize). It is possible that this initial configuration includes some agents in arbitrary $\resetting$ states. The following Corollary will let us only have to reason about starting from a fully computing configuration (so we can essentially assume WLOG that our initial configuration is fully computing, since we will quickly leave the $\resetting$ states).
    
    \begin{corollaryrep}
    \label{cor:reset-fully-computing}
        Starting from any configuration, we reach either an awakening configuration or a fully computing configuration after $O(\log n + \delaymax)$ parallel time with probability at least $1-O(1/n)$.
    \end{corollaryrep}
    
    \begin{proof}
        If at any point we enter a partially-triggered configuration, then the result follows from Theorem~\ref{thm:reset-time}. Now we will show that if we do not enter a partially-triggered configuration, we must reach some fully computing configuration.
        
        We consider a non-fully-computing configuration that is also not triggered (some agents are in $\resetting$ state, but none have $\resetcount=\resetcountmax$).
        In this case, following the proof of Lemma~\ref{lem:fully-propagating-to-fully-dormant}, every agent will become dormant (or computing) after $O(n(\log n + \resetcountmax))$ interactions with probability $1-O(1/n)$. Then as in the proof of Theorem~\ref{thm:reset-time}, some agent will become computing after $O(n\delaymax)$ interactions. Once some agent is computing, the process for the whole population to become computing is a superepidemic, so by Corollary~\ref{cor:epidemic:upper:tail}, this takes at most $3n\ln n$ with high probability $1-O(1/n^2)$.
    \end{proof}
    
\end{toappendix}

%% file: Sections/linear-time-linear-state-silent.tex
\section{Linear-time, linear-state, silent protocol}
\label{sec:silent-linear-time-state}

In this section, we present a silent self-stabilizing ranking protocol, \silentlinearTimeStateProtocol, which achieves asymptotically optimal $O(n)$ time and state complexity. Like \OriginalSSLE, there will be a unique stable and silent configuration where every agent has a unique rank, but now a rank collision will trigger our \propagateresetprotocol, causing the entire population to reset.
The key idea behind \silentlinearTimeStateProtocol\ is to add a large delay $\delaymax = \Theta(n)$ in the \propagateresetprotocol, which will ensure that the entire population is dormant for long enough to do a simple slow leader election 
via $L,L \to L,F$, where all agents set themselves to $L$ upon entering the \resetting\ role.
Thus after the population has undergone a reset, we have a unique leader with high probability. After this reset, we do a linear-time leader-driven ranking, where the ranks correspond to nodes in a full binary tree rooted at the leader.
In this ranking algorithm,
each agent that has been assigned a rank 
(starting with the leader) 
assigns ranks directly to 0, 1, or 2 other agents 
(depending on its number of children in the tree).

In more detail,
each agent can be classified into three roles: 
$\settled$, 
$\unsettled$, and 
$\resetting$.
A $\settled$ agent has the field $\rank \in \{1,2,...,n\}$. 
On the other hand, an $\unsettled$ agent has no rank, and it waits for the assignment of a rank from $\settled$ agents.

We use the subprotocol \propagateresetprotocol\ described in Section~\ref{sec:resetting-subprotocol} to reset each agent when detecting errors. 
For \silentlinearTimeStateProtocol,
the resetting process is triggered under two different situations. 
1) Two $\settled$ agents have an identical rank. 
The rank conflict can be detected when the two agents interact. 
2) An $\unsettled$ agent does not get its rank after $\Theta(n)$ interactions.

During the dormant phase of \propagateresetprotocol,
lasting for $\Theta(n)$ time in this protocol,
we do slow leader election via $L,L \to L,F$.
Upon awakening (calling $\resetprotocol$), 
the (likely unique) leader $L$ is
$\settled$ with $\rank=1$
and followers $F$ are $\unsettled$.
Thus, after resetting, with high probability there will be exactly one $\settled$ agent with $\rank = 1$, and all the other agents are $\unsettled$. 
The $\settled$ agent will act as a leader to assign ranks to all $\unsettled$ agents in the following way.
At this point the protocol executes an initialized ranking algorithm, similar to others in the renaming literature~\cite{alistarh2010fast, alistarh2014balls}.
Intuitively, a full binary tree forms within the population.
Each $\settled$ agent recruits at most two $\unsettled$ agents,
assigning them ranks based on its own to guarantee uniqueness.
The children of rank $i$ are $2i$ and $2i+1$; in other words if an agent's rank has binary expansion $s$, its childrens' ranks have binary expansions $s0$ and $s1$.
Since each agent knows the exact population size, each knows whether its rank corresponds to a node with 0, 1, or 2 children in the full binary tree with $n$ nodes. See Figure~\ref{fig:binary-tree} for an example.
This process clearly terminates when all agents are recruited and become settled into different ranks.

\begin{figure}[h]
    \centering
    \ifarticle
        \includegraphics[width=3.8in]{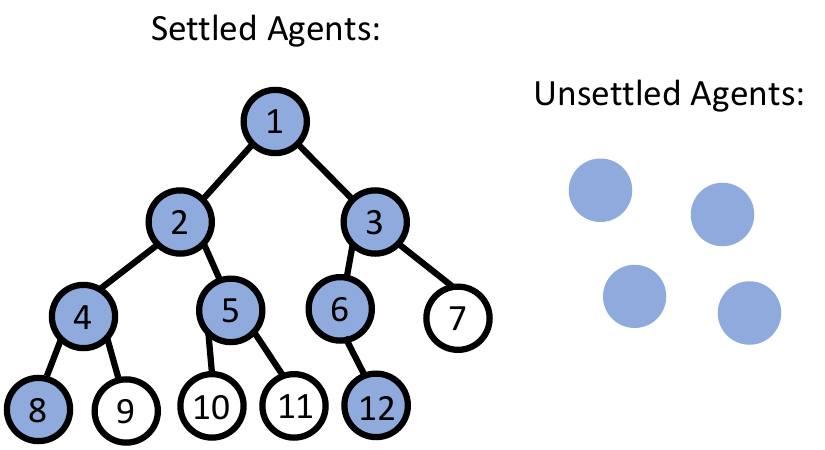}
    \else
        \includegraphics[width=3.0in]{Sections/figures/BinaryTreeFig.pdf}
    \fi
    \caption{An example of the rank assignment in \silentlinearTimeStateProtocol\ with $n=12$ agents. There are $8$ settled agents on the left (blue circles), with ranks given by the numbers. There are 4 ranks in the binary tree left to be filled by the unsettled agents, when they interact with the settled agents with ranks 3,4 or 5. Lemma~\ref{lem:binary-tree-rank-assignment} shows this process completes in expected $\Theta(n)$ time.}
    \label{fig:binary-tree}
\end{figure}

\silentlinearTimeStateProtocol\ takes linear time by the following high-level argument:
If there is a rank collision, this is detected in $O(n)$ time.
If any agent remains $\unsettled$ without a rank,
this is detected via counting up to $\errorcount$ in $O(n)$ time.
Either of these triggers a call to $\propagateresetprotocol$.
Upon exiting and reaching a fully computing configuration,
taking $O(n)$ time by Corollary~\ref{cor:reset-fully-computing}, we have the leader-driven ranking protocol analyzed in Lemma~\ref{lem:binary-tree-rank-assignment}.
There is a constant probability the slow leader election fails (i.e., we end up with multiple leaders),
but the expected number of times we must repeat this process before getting a unique leader is constant.
The fact that this ranking protocol is $O(n)$ time follows by analyzing each level of the binary tree created across the population: 
each level takes time proportional to the number of nodes in the level,
whence the time is proportional to the size of the tree, i.e., $O(n)$.



\begin{protocol}[H]
\caption{\silentlinearTimeStateProtocol, for initiator $a$ interacting with responder $b$
\\
\textbf{Fields:} $\roleself \in \{ \settled, \unsettled, \resetting\}$
\\
If $\roleself = \settled$, $\rank\in\{1,...,n\}$, $\children \in \{0, 1, 2\}$
\\
If $\roleself = \unsettled$, $\errorcount\in \{0,1,...,\errorcountmax\}$
\\
If $\roleself = \resetting$, $\leader \in \{L, F\}$, $\resetcount \in \{1,\ldots,\resetcountmax\}$, $\delay \in \{0, 1, \ldots, \delaymax = \Theta(n)\}$
}
\label{algo:silent-linear-time-state}
\begin{algorithmic}[1]
\If{$a.\role = \resetting$ or $b.\role = \resetting$}
    \State{execute \propagateresetprotocol(a,b)}
    \If{$a.\leader = L$ and $b.\leader = L$}
        \State{$b.\leader \gets F$}
        \label{line:slow-leader-election}
    \EndIf
\EndIf
\If{$a.\role = b.\role = \settled$ and $a.\rank = b.\rank$}
    \State{$a.\role,b.\role \gets \resetting, 
            \quad 
            a.\resetcount,b.\resetcount \gets \resetcountmax$}
    \State{$a.\leader, b.\leader \gets L$}
\EndIf

\For{$(i,j)\in\{(a,b),(b,a)\}$}
\label{line:binary-tree-rank-assignment-start}
    \If{$i.\role = \settled$, $j.\role = \unsettled$, $i.\children < 2$, and $2 \cdot i.\rank + i.\children < n$}
        \State{$j.\role \gets \settled, \quad j.\children \gets 0$}
        \State{$j.\rank \gets 2 \cdot i.\rank + i.\children$}
        \Comment{$j$ becomes a child node of $i$}
        \State{$i.\children \gets i.\children + 1$}
        \label{line:binary-tree-rank-assignment-end}
    \EndIf
\EndFor

\For{$i \in \{a,b\} $}
    \If{$i.\role = \unsettled$}
        \State{$i.\errorcount \gets \max(i.\errorcount - 1, 0)$}
        \If{$i.\errorcount = 0$}
            \State{$a.\role, b.\role \gets \resetting, 
            \quad 
            a.\resetcount,b.\resetcount \gets \resetcountmax$}
            \State{$a.\leader, b.\leader \gets L$}
        \EndIf
    \EndIf
\EndFor

\end{algorithmic}
\end{protocol}

\begin{protocol}[H]
\caption{\resetprotocol(a)\ for \silentlinearTimeStateProtocol, for agent $a$. 
\\
(Called in line~\ref{algo:execute-reset} of \propagateresetprotocol.)}
\begin{algorithmic}[1]
    \If{$a.\leader = L$}
    \State $a.\roleself \gets \settled,
    \quad 
    a.\rank\gets 1,
    \quad
    a.\children\gets 0$
    \EndIf
    \If{$a.\leader = F$}
    \State{$a.\roleself \gets \unsettled,
    \quad
    a.\errorcount \gets \errorcountmax$}
    \EndIf
\end{algorithmic}\label{algo:reset-silent-linear-time-state}
\end{protocol}

\begin{toappendix}
\subsection{Proofs for \silentlinearTimeStateProtocol}
\label{sec:silent-protocol-proofs}

We first consider the \emph{binary-tree rank assignment} process, which is given by lines~\ref{line:binary-tree-rank-assignment-start} to \ref{line:binary-tree-rank-assignment-end} of \silentlinearTimeStateProtocol, starting from a configuration with one agent $a$ with $a.\role = \settled, a.\leader = L$, $a.\children = 0$ and all other agents $b$ with $b.\role = \unsettled$ (we ignore the field $\errorcount$ in this analysis). In other words, we are starting from a single leader $a$, who builds a binary tree, rooted at $a$ that assigns roles to the entire population.

\begin{lemma}
\label{lem:binary-tree-rank-assignment}
The binary-tree rank assignment process takes expected $O(n)$ time, starting from a single leader and $n-1$ $\unsettled$ agents.
\end{lemma}

\begin{proof}
    We consider the time taken to assign all nodes at level $d$ of the tree,
    assuming all nodes at level $d-1$ have been assigned, showing that this time is $O(2^d)$ 
    (i.e., proportional to the number $2^d$ of nodes at level $d$).
    Thus, even if we consider the stochastically dominating process in which no agent is assigned at level $> d$ until all agents at level $d$ have been assigned,
    the time to complete the tree is 
    $O\left( \sum_{d=1}^{\log n} 2^d \right) = O(n)$.
    
    Assume that level $d-1$ is completely assigned,
    and let $i$
    be the number $i$ of nodes still unassigned at level $d$, with $0 < i \leq 2^d$.
    We will now estimate the probability that the next interaction assigns a new node at level $d$. The count of agents at level $d-1$ with $\children < 2$ is at least $\frac{i}{2}$. To estimate the number of $\unsettled$ agents, for each of the $i$ unassigned nodes at level $d$, we consider the eventually subtree rooted at that node. Each such subtree contains 
    \[
    \sum_{j=0}^{\log n - d} 2^j = 2\cdot 2^{\log n - d} - 1 \leq 2^{\log n - d}
    \]
    nodes. Since the root of the subtree has not been assigned yet,
    none of the nodes in the subtree have been assigned either.
    Thus the number of $\unsettled$ agents is at least $i\cdot n\cdot2^{-d}$. The probability of an $\unsettled$ agent meeting an agent at level $d-1$ with $\children < 2$ is at least 
    \[
    \frac{(i/2)(i\cdot n \cdot 2^{-d})}{\binom{n}{2}}
    \sim \frac{i^2}{n 2^d}.
    \]
    
    The total number of interactions to assign all nodes at level $d$ is stochastically dominated by a sum $T = \sum_{i=1}^n T_i$ of independent geometric random variables, where $T_i$ is a geometric random variable with success probability $\frac{i^2}{n 2^d}$.
    We can then compute
    \[
    \IE[T] = \sum_{i=1}^n \frac{n 2^d}{i^2} \leq n2^d \sum_{i=1}^\infty \frac{1}{i^2} = O(n 2^d).
    \]
    Since $T$ is counting interactions, the expected parallel time to assign all nodes at level $d$ is $O(2^d)$, as desired.
\end{proof}

\begin{lemma}
\label{lem:single-leader}
An awakening configuration has a single leader with constant probability.
\end{lemma}

\begin{proof}
We first analyze the simple leader election process executed in line~\ref{line:slow-leader-election} of \silentlinearTimeStateProtocol.
The process $L,L \to L,F$ completes in number of interactions $\sum_{i=2}^{n} T_i$,
where $T_i$ is a random variable representing the number of interactions to get from $i$ leaders down to $i-1$.
This is geometric with probability of success $\frac{\binom{i}{2}}{\binom{n}{2}} = \frac{i(i-1)}{n(n-1)}$.
Thus the expected number of interactions is
\[
\sum_{i=2}^{n} \frac{n(n-1)}{i(i-1)} 
= (n-1) n \sum_{i=2}^{n} \frac{1}{i-1} - \frac{1}{i} 
= n(n-1) \left[1 - \frac{1}{n} \right] 
\sim n^2,
\]
so the expected parallel time is $\sim n$.

For an awakening configuration to have a single leader, we need this leader election process to finish before the first agent sets $\delay = 0$, which takes $O(n)$ interactions, and will take $O(n)$ parallel time with high probability by standard Chernoff bounds.
By Markov's inequality,
for any $\alpha > 0$,
the probability that the leader election takes longer than $\alpha n$ time is a most $1 / \alpha$.
Thus we have a constant probability of successfully electing a single leader before the first dormant agent counts up to $\delaymax$, 
where the precise constant depends on the choice of $\delaymax = \Theta(n)$.
\end{proof}

\end{toappendix}

The required proofs for \silentlinearTimeStateProtocol\ are given in Section~\ref{sec:silent-protocol-proofs} and yield the following main results:

\begin{theoremrep}
\label{thm:silent-linear-time}
\silentlinearTimeStateProtocol\ is a silent protocol that solves self-stabilizing ranking with $O(n)$ states and $O(n)$ expected parallel time.
\end{theoremrep}
\begin{proof}
By Corollary~\ref{cor:reset-fully-computing}, we only have to consider initial configurations that are fully computing or awakening. First consider any fully computing configuration, and argue until we either reach a partially triggered configuration or the unique stable configuration. After expected $O(n)$ time, any $\unsettled$ agents have had enough interactions to count down to $\errorcount = 0$. By standard Chernoff bounds, this is true after $O(n \log n)$ time with high probability. Thus we have either reached a partially triggered configuration, or have only $\settled$ agents left. Now if this configuration is not the unique silent, stable configuration, there is a rank collision between two $\settled$ agents, which again will be detected in expected $O(n)$ time and $O(n\log n)$ time with high probability. After reaching a partially triggered configuration, by Theorem~\ref{thm:reset-time}, we will then be in an awakening configuration after $O(n)$ time with high probability.

It now remains to analyze the process starting from an awakening configuration. Define an epoch to be the sequence of configurations, starting from an awakening configuration, that ends at either the unique stable configuration or another partially triggered configuration. If it ends in the unique stable configuration, we call the epoch successful. We will now show that an epoch is successful with constant probability, the expected time for an epoch is $O(n)$.

For an epoch to be successful, we first require the awakening configuration to have a single leader, which is true with constant probability by Lemma~\ref{lem:single-leader}. Then given a single leader in an awakening configuration, we can analyze the process of the leader assigning ranks to the rest of the population with the analysis of the binary-tree rank assignment process in Lemma~\ref{lem:binary-tree-rank-assignment}, which shows that it takes expected $O(n)$ time. By Markov's inequality, it will also take $O(n)$ time with constant probability, which will be before any of the $\unsettled$ agents have had enough interactions to reach $\errorcount = 0$ (again by standard Chernoff bounds on the number of interactions for each agent, where the exact probability depends on the initial value $\errorcountmax$). As a result, with constant probability we finish the rank assignment, and reach the unique stable configuration with all $\settled$ agents, having unique ranks.

The expected $O(n)$ time for an epoch follows immediately from Lemma~\ref{lem:binary-tree-rank-assignment} and Theorem~\ref{thm:reset-time}. Now the number of required epochs is a geometric random variable with constant expectation, so the total expected time is still $O(n)$.

In analyzing the state set, note that each role using $O(n)$ states. The $\settled$ agents have $O(n)$ states for $\rank$, the $\unsettled$ agents have $O(n)$ states for $\errorcount$, and the $\resetting$ agents have $O(\resetcountmax + \delaymax) = O(n)$ states.
The total state set is the sum of each of these disjoint roles, giving $O(n)$ total states.
\end{proof}

\begin{corollaryrep}
\silentlinearTimeStateProtocol\ takes $O(n \log n)$ time with high probability.
\end{corollaryrep}

\begin{proof}
By Theorem~\ref{thm:silent-linear-time}, \silentlinearTimeStateProtocol\ takes expected $O(n)$ time, so by Markov's inequality it takes $O(n)$ time with constant probability. Now consider epochs of fixed length $O(n)$ time, where an epoch is successful if we reach the unique stable configuration by the end of the epoch. If the epoch is not successful, then we can consider the final configuration as the initial configuration, because the protocol is self-stabilizing. Thus the analysis of Theorem~\ref{thm:silent-linear-time} holds for each epoch. The total number of epochs required is then a geometric random variable with constant success probability, which is $O(\log n)$ with high probability ($1 - O(1/n^c)$ for any desired $c$). Since these epochs were defined to use a fixed amount of $O(n)$ time, the high probability bound is $O(n \log n)$ time.
\end{proof}


%% file: Sections/log-time-bounded-state.tex
\section{Logarithmic-time protocol}
\label{sec:log-time-protocol}

In this section, we show a protocol solving SSR, and thus SSLE, in optimal $O(\log n)$ expected time, using 
a ``quasi-exponential'' number of states: $\exp(O(n^{\log n} \cdot \log n))$.
%
Observation~\ref{obs:linear-time-lower-bound-silent} shows that to achieve sublinear time, the protocol necessarily must be non-silent: 
agents change states forever.

\subsection{Overview}

Intuitively,
\logTimeProtocol\ works as follows.
Each agent has a field \name, a bitstring of length $3\log_2 n$. The $n^3$ possible values ensure that if all agents pick a \name\ randomly,
with high probability, there are no collisions.
The set of all \name\ values in the population is propagated by epidemic in $O(\log n)$ time in a field called \Met\ 
(which has an exponential number of possible values).
Agents 
{update their \rank\ (a write-only output field) only when} 
their \Met\ field has size $n$;
in this case the agent's \rank\ is its \name's lexicographical order in the set \Met.

One source of error is that we can start in a configuration with a ``\emph{ghost name}'':
a name that is in the \Met\ set of some agent,
but that is not the \name\ of any agent.
If there are no collisions among actual \name's,
this error is easy to handle:
eventually we will have $|\Met| > n$,
indicating that there is a ghost name, 
triggering \propagateresetprotocol.\footnote{
    Eventually we will introduce a tree data structure that also has all the names. 
    However, it is necessary to keep a separate set \Met\ of names for the following reason.
    The set \Met\ is propagated in time $O(\log n)$, whereas in slower variants of our algorithm, the tree takes too long to collect the names.
    For example, in the $O(\sqrt{n})$ time (and uses less memory) variant, the tree takes time $\Omega(n)$ to populate with all names.
}

The main challenge is then to detect name collisions.
\logTimeProtocol\ calls a subroutine \detectcollision\ that detects whether two agents have the same \name.
If so, we call the same subroutine \propagateresetprotocol\ used in \silentlinearTimeStateProtocol,
now with $\delaymax = \Theta(\log n)$ rather than $\Theta(n)$ as in \silentlinearTimeStateProtocol.
Upon awakening from \propagateresetprotocol,
agents pick a new name randomly.
They use their dormant time,
while still in role \resetting,
but with $\resetcount=0$
while counting \delay\ down to 0,
to generate random bits to pick a new random name.

The bulk of the analysis is in devising an $O(\log n)$ time protocol implementing \detectcollision.
The rest of the protocol outside of \detectcollision\ is silent:
once the protocol stabilizes, no \name\ or \rank\ field changes.
Indeed, we can implement a silent protocol on top of this scheme if we are content with $\Theta(n)$ time:
\detectcollision\ can be implemented with the simple rule that checks whether the \name\ fields of the two interacting agents are equal,
i.e., direct collision detection.
The challenge, therefore, is in implementing \detectcollision\ in sublinear time by \emph{indirectly} detecting collisions, 
without requiring agents with the same name to meet directly.
By Observation~\ref{obs:linear-time-lower-bound-silent} 
any method of doing this will necessarily be non-silent.

The protocol \detectcollision\ is parameterized to give a tradeoff between stabilizing time and state complexity.
For instance, there is a $O(\sqrt{n})$ time protocol that uses a data structure with $kn$ bits for a parameter $k$, i.e., $2^{kn}$ possible values.
Of course, all of the schemes use at least exponential states, 
since the field \Met\ has $\approx n^{3n}$ possible values.
However, the faster schemes will use even more states than this,
and their analysis is more complex.
This is discussed in more detail in Section~\ref{subsec:fast-collision-detection}.

The proofs in Section~\ref{sec:log-time-top-level-proofs} shows this top-level protocol works as intended, once we have correct and efficient collision detection from \detectcollision.

\begin{protocol}[H]
\caption{\logTimeProtocol, for agent $a$ interacting with agent $b$. 
\\
\textbf{Fields:} $\roleself \in \{ \collecting, \resetting\}$, $\name \in \Q$
\\
If $\roleself = \collecting$, $\rank \in \{1,\dots,n\}$,
        $\Met \subseteq \Q, |\Met| \leq n$,
        other fields from \detectcollision
\\
If $\roleself = \resetting$, $\resetcount \in \{1,\ldots,\resetcountmax\}$, $\delay \in \{0, 1, \ldots, \delaymax = \Theta(\log n)\}$
}
\label{algo:lineartime_bounded}
\begin{algorithmic}[1]
    \If{$a.\role = b.\role = \collecting$}
    
        \If{\detectcollision($a$,$b$) or $\lvert a.\Met \cup b.\Met \rvert>n$}
        \label{algo:lineartime:detect-error}
        \State{$a.\role,b.\role \gets \resetting, 
                    \quad 
                    a.\resetcount,b.\resetcount \gets \resetcountmax$}
            \label{algo:lineartime:trigger-reset1}
        \Else
        \State $a.\Met,b.\Met \gets a.\Met \cup b.\Met$ 
        \qquad

        \If{$\lvert a.\Met \cup b.\Met \rvert = n$}
        \Comment{do not set rank until all names collected}
        \For{$i\in\{a,b\}$}
            \State{$i.\rank \gets$ lexicographic order of $i.\name$ in $\Met$ }
            \label{line:set-rank-lexicographic}
        \EndFor
        \EndIf
        \EndIf

    \Else
        \Comment{some agent is $\resetting$}
    \State execute \propagateresetprotocol($a$,$b$) 
       \For{$i \in \{a,b\}$ such that $|i.\resetcount| > 0$}
            \State{$i.\name \gets \varepsilon$}
            \Comment{clear names while propagating the reset signal}
        \EndFor 
        \For{$i \in \{a,b\}$ such that $i.\resetcount = 0$ and $|i.\name|<3\log_2 n$}
            \State{append a random bit to $i.\name$}
            \Comment{can be derandomized, see Section~\ref{sec:derandomization}}
            \label{line:random-bit}
        \EndFor
    \label{algo:lineartime:trigger-reset2}
    \EndIf
\end{algorithmic}
\end{protocol}

\begin{protocol}[H]
\caption{\resetprotocol(a)\ for \logTimeProtocol, for agent $a$.
\\
(Called in line~\ref{algo:execute-reset} of \propagateresetprotocol.)}
\begin{algorithmic}[1]
\opt{inline,append,strip}{
    \State $\roleself \gets \collecting$
    \State $\Met\gets\{\name\}$
}
\end{algorithmic}\label{algo:reset-log-time}
\end{protocol}



\begin{toappendix}
\subsection{Proofs for \logTimeProtocol}
\label{sec:log-time-top-level-proofs}
    
We will first prove a series of Lemmas about the behavior of \propagateresetprotocol, which will show that if $\detectcollision$ works as intended to detect collisions in $O(T_H)$ time, then $\logTimeProtocol$ will solve self-stabilizing ranking in $O(T_H)$ time.

We call a configuration \emph{non-colliding} if all agents have distinct names ($a.\name \neq b.\name$ for all $a,b\in \Agents$) and $|a.\name| = 3\log_2 n$ for all $a\in\Agents$. This ensures any $\resetting$ agents have generated enough random bits to have picked a new name, and all new names are unique. We first reason about awakening configurations, which we show are non-colliding with high probability.

\begin{lemmarep}
\label{lem:unique-names}
With high probability $1-O(1/n)$, from a partially triggered configuration, we reach an awakening, non-colliding configuration.
\end{lemmarep}

\begin{proof}
By Theorem~\ref{thm:reset-time}, we reach an awakening configuration after $O(\delaymax)$ time.
Also, for $\Theta(\delaymax)$ time, all agents will be dormant (with $\resetcount = 0$ and decrementing $\delay$). By standard Chernoff bounds for appropriate choice of constant $\delaymax = \Theta(\log n)$, all agents will be dormant for long enough to generate all $3\log_2 n$ bits of a new name, with high probability.

In this case, each agent has a uniform random name out of all $O(n^3)$ possible bit strings. The probability of any arbitrary pair of agents choosing the same name is $O(1/n^3)$, by union bound over all pairs of agents, the probability of name collision is $O(1/n)$. Thus with high probability $1-O(1/n)$, the awakening configuration we reach is non-colliding.
\end{proof}

Now we can appeal to the roll call process to show that from an awakening, non colliding configuration, we quickly get to unique ranks.

\begin{lemmarep}
\label{lem:unique-ranks-from-unique-names}
From an awakening, non-colliding configuration, we reach a configuration with unique ranks in $O(\log n)$ time with high probability $1-O(1/n)$.
\end{lemmarep}

\begin{proof}
The agents update their field $\Met$ by taking unions, so the process of all agents getting a complete $\Met$ is exactly the roll call process. By Lemma~\ref{lem:parallel:epidemics}, this finishes within $3 \ln n$ parallel time with high probability $1-O(1/n)$. Once every agent has all $n$ unique names in their $\Met$, they choose unique ranks by the lexicographic order of their name, in line~\ref{line:set-rank-lexicographic} of \logTimeProtocol.
\end{proof}

Note that if we were using simple direct interactions to detect name collisions, then this unique rank configuration would be silent and stable. For the actual algorithm, we must later prove a safety condition, that $\detectcollision$ will never falsely detect a collision in this stable configuration, shown as Lemma~\ref{lem:safety-from-woke-configurations}.

There are now two different errors that we must quickly detect. One is name collisions, which we will handle in the next section with $\detectcollision$. The other is ``ghost names''.
We say a configuration is \emph{ghostly} if some $\Met$ set contains a ghost name, i.e., if
\[
    \bigcup_{a\in\Agents}
    a.\Met
    \not\subseteq
    \{b.\name \mid b\in\Agents\}
    \enspace.
\]
If the configuration is non-colliding, then it is easy to detect ghost names, because the full set of names will be larger than $n$. Thus we can again appeal to the roll call process to show we quickly detect an error in this case.

\begin{lemmarep}\label{lem:ghostly}
    Starting from a ghostly and non-colliding configuration, we reach a partially triggered configuration after $O(\log n)$ time with high probability $1-O(1/n)$.
\end{lemmarep}

\begin{proof}
The agents update their field $\Met$ by taking unions, so the process of all agents getting a complete $\Met$ is exactly the roll call process. By Lemma~\ref{lem:parallel:epidemics}, this finishes within $3 \ln n$ parallel time with high probability $1-O(1/n)$. Because the configuration is ghostly and noncolliding, the total number of distinct names in the rosters is $>n$. Once some agent has $\Met > n$, they will become triggered by line~\ref{algo:lineartime:trigger-reset1} of \logTimeProtocol.
\end{proof}
\end{toappendix}

\subsection{Fast Collision Detection}
\label{subsec:fast-collision-detection}

In \logTimeProtocol, both \propagateresetprotocol\ and filling all agents' \Met\ take $O(\log n)$ time, so the time bottleneck is waiting to detect a name collision. If we simply wait for two agents with the same name to meet to detect a collision, this will take $\Theta(n)$ time in the worst case, which would give a $\Theta(n)$ time silent algorithm.

The goal of \detectcollision\ is to detect these names collisions in sublinear time. Because of the lower bound of Observation~\ref{obs:linear-time-lower-bound-silent}, this protocol must not be silent. 
\detectcollision\ will have to satisfy two conditions. In order to allow $O(\log n)$ time convergence, from any configuration with a name collision, some agent must detect this collision in $O(\log n)$ time to initiate \propagateresetprotocol. 
Second, to ensure the eventual ranked configuration is stable, it must satisfy a safety condition where from a configuration with unique names, 
no agent will ever think there is a name collision.\footnote{
    The initial configuration could have unique names, but with auxiliary data adversarially planted to mislead agents into believing there is a name collision, triggering a reset.
    So the actual safety condition is more subtle and involves unique-name configurations reachable only after a reset.
}

As a warm-up to the full $O(\log n)$-time protocol of \detectcollision, consider the following simpler $O(\sqrt{n})$-time protocol. 
Each agent keeps a dictionary keyed by names of other agents they have encountered in the population. Whenever agents $a$ and $b$ meet, they generate a random shared value \sync\ $\in \{1,\ldots,k\}$, which $a$ stores in its dictionary keyed by the name of $b$, and $b$ stores in its dictionary keyed by the name of $a$. If the two agents disagree on this \sync\ value at the beginning of an interaction, they declare a name collision.

From a configuration with two agents $a$ and $a'$ who share the same name, within $O(\sqrt{n})$ time, some agent $b$ will interact with both $a$ and $a'$ 
(assume $b$ first interacts with $a$, then $a'$). 
With probability $1 - \frac{1}{k}$, the sync value that $b$ generates with $a$ will disagree with the sync value that $a'$ has with $b$. Thus when $b$ then meets $a'$, it is able to detect a name collision.
Also note that from a configuration with unique names, an invariant is maintained that all pairs of agents agree on their corresponding \sync\ values, giving the required safety property.

The actual protocol \detectcollision\ is a generalization of this idea. The agents now store a more complicated data structure: a tree whose nodes are labelled by names. 
See Figure~\ref{fig:trees-example} for an example.
The root is labelled by the agent's own name, and every root-to-leaf path is \emph{simply labelled}, meaning that each node on the path contains unique names
(it is permitted for the same name to appear on multiple nodes in the tree, 
but neither of these nodes can be an ancestor of the other).
Each edge is labelled by a \sync\ value. 
The intuition is that these paths correspond to \emph{histories}: 
chains of interactions between agents, 
where the \sync\ values on the edges were generated by the interaction between that pair of agents. 
For instance, 
if $a$ has a path $a \xrightarrow{3} b \xrightarrow{5} c \xrightarrow{7} d$, 
the interpretation is that when $a$ last met $b$, 
$a$ and $b$ generated sync value 3, 
and in that interaction, 
$b$ told $a$ that when $b$ last met $c$, 
$b$ and $c$ generated sync value 5, 
and in that interaction, 
$c$ told $b$ that when $c$ last met $d$, 
$c$ and $d$ generated sync value 7.
In particular, it could be that $c$ and $d$ have interacted again, 
generating a different sync value than 7, before $a$ and $b$ interact,
but $b$ has not heard about that interaction.
See Fig.~\ref{fig:trees-example} for an example showing how this information is built up.

\begin{protocol}[H]
\caption{\detectcollision(a,b)\ for \logTimeProtocol, for agent $a, b$.
\\
\textbf{Fields:} \tree: depth \depth, root labelled \name, other nodes have $\node.\name\in\Met$. Edges have $\edge.\sync\in\{1,\ldots,\synccount = \Theta(n^2)\}$ and $\edge.\edgetimer\in\{0,\ldots,\edgetimercount\}$.
The parameter $T_H = \Theta(H \cdot n^{1/(H+1)})$ for $H = O(1)$ and $T_H = \Theta(\log n)$ for $H = \Theta(\log n)$ (we need $T_H = \Theta(\tau_{H+1})$ as in Lemmas~\ref{lem:bounded-epidemic-constant} and \ref{lem:bounded-epidemic-log-n}).
}
\begin{algorithmic}[1]
    \For{$(i,j)\in\{(a,b),(b,a)\}$}
        \For{every path $(i.e_1,\ldots,i.e_p)$ in $i.\tree$ with $i.e_1.\edgetimer, \ldots, i.e_p.\timer > 0$ and last node $v$ with $v.\name = j.\name$}
        \label{line:check-positive-timers}
        \Comment{All of $i$'s histories about $j$ that aren't outdated}
            \If{$\pathconsistencyprotocol(j,(i.e_1,\ldots,i.e_p))=\inconsistent$}
            \label{line:check-path-consistency}
                \State{Return \true}
                \Comment{collision detected}
            \EndIf
        \EndFor
    \EndFor
    \State{$x\gets$ chosen uniformly at random from $\{1,\ldots,\synccount\}$}
    \Comment{Choose new sync value}
    \For{$(i,j)\in\{(a,b),(b,a)\}$}
        \Comment{Update trees to share new information}
            \If{$i.\tree$ has node $v$ at depth 1 with $v.\name = j.\name$}
                \State{Remove the subtree rooted at $v$ from $i.\tree$}
            \EndIf
            \State{Add $j.\tree$ (to depth $H-1$) as a subtree of $i.\tree$ via new edge $e$ from the root}
            \label{line:copy-subtree}
            \State{$e.\sync \gets x$, $e.\edgetimer\gets \edgetimercount$}
    \EndFor
    \For{$i \in \{a,b\}$}
    \Comment{Keep the trees simply labelled}
        \State{remove from $i.\tree$ all subtrees with root labelled with $i.\name$}
    \EndFor
    \For{each edge $e$ in $a.\tree$ and $b.\tree$}
        \State{$e.\edgetimer\gets \max(e.\edgetimer - 1$,0)}
    \EndFor
    \State{Return \false}
    \Comment{no collision detected}
\end{algorithmic}\label{algo:detectcollision}
\end{protocol}

The $O(\sqrt n)$ time algorithm above can be thought of as a tree of depth 1, where each agent stores only the names and \sync\ values of the agents it has directly interacted with. The general algorithm has a tree of depth $H$, which allows agents to hear about other agents' \sync\ values through longer chains of interactions. 
In line~\ref{line:check-path-consistency} 
of \detectcollision, 
each agent checks any paths ending at the name of the other agent 
(the additional fields $\edge.\timer$ are a technicality to handle certain adversarial initial conditions, see Lemma~\ref{lem:safety-from-all-configurations}). 
Intuitively, they require the other agent to show information that is logically consistent with this path, 
formalized in the conditions of \pathconsistencyprotocol. 
To detect a name collision between agents $a$ and $a'$, 
it will now suffice for some agent $b$ to have heard about agent $a$ before meeting $a'$. 
With constant probability, 
the duplicate agent $a'$ will not have any \sync\ values that are logically consistent with this path, 
and $b$ will declare a collision. 
Allowing longer paths decreases the time it takes for this information to travel between $a$ and $a'$. 
Because the paths that spread information in the epidemic process have length at most $O(\log n)$ with high probability (see Lemma~\ref{lem:bounded-epidemic-log-n}), 
once we take $H=O(\log n)$, 
in the $O(\log n)$ time it would take for an epidemic starting at $a$ to reach $a'$, 
some agent will detect a collision in this way.

\begin{protocol}[H]
\caption{\pathconsistencyprotocol(j,P)\ for \detectcollision, for agent $j$ verifying path $P=(i.e_1,\ldots,i.e_p)$}
\begin{algorithmic}[1]
    \State{$q \gets \min\{q' \mid \exists (j.e_p,\ldots,j.e_{q'}) $ in $j.\tree \}$}
    \Comment{$(j.e_p,\ldots,j.e_{q})$ is a root-to-leaf path}
    \For{edge $j.e \in(j.e_p,\ldots,j.e_q)$}
        \If{$j.e.\sync = i.e.\sync$}
            \State{Return $\true$}
        \EndIf
    \EndFor
    \State{Return \inconsistent}
\end{algorithmic}\label{algo:pathconsistency}
\end{protocol}

\newcommand{\rootoftree}[1]{
  \begin{forest}
  for tree={circle,draw, l sep=20pt} [#1]
  \end{forest}
}
\newcommand{\newinfo}[1]{{\color{red} #1}}

\begin{figure}


\input{Sections/figures/path1}
\hspace{-0.05cm}
\ifarticle
    \rule[-8.7cm]{0.05cm}{17.5cm} 
\else
    \rule[-7.5cm]{0.05cm}{15cm} 
\fi
\hspace{0.05cm}
\input{Sections/figures/path2}

    \caption{ 
    \ifarticle
        \footnotesize 
    \fi
    Example executions building up trees in agents, starting from a ``clean'' configuration with singleton trees. 
    {\color{red} Red} sync values are those that are newly generated or communicated in the preceding interaction.
    As an example of how agents check for consistency, when $a$ and $d$ interact, before updating their trees, $d$ checks any paths $p$ that end with $a$ (here there's just one, $d \xrightarrow{3} c \xrightarrow{2} b \xrightarrow{1} a$) 
    against $a$'s corresponding path, which is $a$'s longest reversed suffix of $p$.
    In the example on the left, $a$'s reverse suffix is $a \xrightarrow{1} b$, with just a single edge that matches the final sync value in this path $p$, so \pathconsistencyprotocol\ will return \true\ after checking the first edge.
    In the example on the right, $a$'s reverse suffix is $a \xrightarrow{7} b \xrightarrow{2} c$. The first edge $a \xrightarrow{7} b$ does not match $d$'s tree, because agents $a$ and $b$ generated the new sync value $7$ in a later interaction. However, in that interaction, $a$ added the edge $b \xrightarrow{2} c$, hearing about the $b$-$c$ interaction with sync value $2$ that matches the path in $d$'s tree. Now \pathconsistencyprotocol\ will return \true\ after checking the second edge.
    }
    \label{fig:trees-example}
\end{figure}

\begin{toappendix}
\subsection{Proofs for \detectcollision}
\label{sec:detect-collision-proofs}

We first argue the safety condition, that configurations with unique names will not return false collisions. We start by arguing about awakening, non-colliding configuration, where every agent will execute \resetprotocol\ and initialize a \tree\ only containing the root with their unique value of \name. Here every value $\node.\name$ in an agent's tree will refer to the name of some unique agent. We will sometimes abuse notation and use $a$ to refer to both an agent and the unique name held by that agent.

We now show that from this initialized setting, no collisions are ever detected, ie. there are no ``false positives''.
See Figure~\ref{fig:trees-example} for an example of why collisions are not detected started from an awakening configuration.

\begin{lemmarep}[Safety after a correct reset]
\label{lem:safety-from-woke-configurations}
An awakening, non-colliding configuration is \emph{safe}:
\detectcollision\ will return \false\ in all future reachable configurations.
\end{lemmarep}


\begin{proof}
Observe that because the configuration is non-colliding, every name used to label the nodes in all trees uniquely correspond to one agent in the population. Also, because the configuration is awakening, every agent will at some point start with an empty tree. Thus every edge $e = (i,j)$ in the tree of any agent corresponds to some interaction between agents $i$ and $j$, who randomly generated the sync value $x = e.\sync$.

Now consider any time we run $\pathconsistencyprotocol(j,P)$, for a simply labelled path $P=(i.e_1,\ldots,i.e_p)$ in the tree of agent $i$, checking this path against agent $j$. 
We prove this returns \true, arguing by induction on the length $p$ of the path.

Base case for paths of length $p=1$: $i.e_1$ corresponds to the last interaction between $i$ and $j$, where they agreed on a sync value $i.e_1.\sync = j.e_1.\sync$. Thus $j$ will still have a matching edge $j.e_1$ in its tree, with the same sync value, causing it to return \true.

Induction step for paths of length $p\geq 2$: Let $i.e_1=(i,a)$, where $a$ is the unique agent with the other name on this edge, and $i.e_p=(b,j)$, where $b$ is the unique agent with the other name on this edge. (Note that for $p=2$, $a = b$.) Consider which of the two pairs $(i,a)$ and $(b,j)$ interacted most recently.

In the first case, assume $(i,a)$ had the more recent interaction than $(b,j)$, at some time $T_{ia}$ before the current time $T_{ij}$. The path $P$ in $i.\tree$ goes from the root to node $a$, so $a$ has a path $P'=(a.e_2, \ldots, a.e_p)$ ending at $j$ that is a length-$(p-1)$ suffix of $P$, which gets copied to $i.\tree$ in this interaction (line~\ref{line:copy-subtree} in \detectcollision).
Now consider a hypothetical interaction scheduled instead between $a$ and $j$ at time $T_{ia}$, when $a$ has the path $P'$ in $a.\tree$.
  By the induction hypothesis,
  in this hypothetical $(a,j)$ interaction, $\pathconsistencyprotocol(j,P')$ would return \true, so $j$ has some reverse suffix $(j.e_p, \ldots, j.e_q)$ with a matching edge $j.e$ where $j.e.\sync = a.e.\sync = i.e.\sync$. Also, because agents $b$ and $j$ do not interact between times $T_{ia}$ and $T_{ij}$, this path will not change in $j.\tree$. Thus $\pathconsistencyprotocol(j,P)$ will also return \true\ at the 
  time $T_{ij}$.

The second case is mostly symmetric.
Assume $(b,j)$ was more recent than $(i,a)$, at some time $T_{bj}$ before the current time $T_{ij}$. Now consider a hypothetical interaction scheduled instead between $i$ and $b$ at time $T_{bj}$. $i$ has the prefix $P'=(i.e_1,\ldots,i.e_{p-1})$ ending at $b$, and $\pathconsistencyprotocol(b,P')$ would return \true, so $b$ has some reverse suffix $S'=(b.e_{p-1}, \ldots, b.e_q)$ with a matching edge $b.e$ where $b.e.\sync = i.e.\sync$. Also, because agents $i$ and $a$ do not interact between times $T_{bj}$ and $T_{ij}$, the path $P'$ will not change in $i.\tree$. In the actual interaction between $b$ and $j$ at time $T$, this reverse suffix $S'$ gets copied to $j.\tree$, so $j$ has a suffix $S = (j.e_p,j.e_{p-1},\ldots,j.e_q)$, which does not change because $j$ has no further interactions with $b$. Thus $\pathconsistencyprotocol(j,P)$ will also return \true\ at the current time, via the edge $e$ in this suffix $S$.
\end{proof}




In a non-self-stabilizing setting, where we could assume the agents are initialized with empty trees, Lemma~\ref{lem:safety-from-woke-configurations} would have been sufficient. However, there is the possibility of adversarial initial conditions that do not have any collisions, but where agents are initialized with inconsistent trees that could falsely detect a collision in a later interaction. It could take up to linear time for the pair of agents with inconsistent data to meet. 

In order to circumvent this issue, we add the field $\edge.\timer$ to each edge. From such an adversarial initial condition, once $\edge.\timer$ has counted down to $0$ for each edge coming from the initial condition and not a true interaction, then by the check on line 2 of \detectcollision, no paths from the initial data will ever get checked. Note that paths whose timers have hit $0$ can still be used as verification in \pathconsistencyprotocol, 
which is essential to ensure correctness, but only if the reverse path still has positive timers.

The following lemma shows that starting from any non-colliding configuration, once all initial timers have run out, either we will have detected a collision, or we will have reached a situation where Lemma~\ref{lem:safety-from-woke-configurations} can apply to give safety in all reachable configurations. It reasons about $O(T_H)$ time, where $T_H$ is the parameter giving the maximum value of the $\edge.\timer$.

\begin{lemmarep}[Safety from all configurations]
\label{lem:safety-from-all-configurations}
Starting from any non-colliding configuration $C$, after $O(T_H)$ time, with probability $1-O(1/n)$, either some agent has become triggered, or we reach a configuration $C_0$ that is \emph{safe}: 
\detectcollision\ will return \false\ in all future reachable configurations.
\end{lemmarep}

\begin{proof}
First we argue that after $O(T_H)$ time, with high probability, all timers that correspond to edges from the initial configuration have reached $0$. 
Define $E$ as the set of all edges in the trees of any agents in the initial configuration. For each agent $a$, define $m(a) = \max\{e.\edgetimer:e \in E \cap a.\tree\} $, i.e. the largest timer value corresponding to an edge that agent has that came from the initial configuration, so is possibly corrupted initial data. Note that because agents share edges in their interaction, and also all edge timers decrement, these counts $m$ update via $m(a), m(b) \gets \max(m(a)-1,m(b)-1, 0)$. So $m$ is a ``propagating'' variable of the exact type analyzed in \propagateresetprotocol. By Lemma~\ref{lem:fully-propagating-to-fully-dormant}, setting $R_{\max} = T_H$, we have $\max_a m(a) = 0$ with high probability in time $O(T_H)$.

Note that before these timers have all run out, it is possible for some agent to return \true\ from \detectcollision, and thus become triggered. Also, if the configuration is ghostly, then some agent will become triggered in $O(\log n)$ time with probability $1-O(1/n)$ by Lemma~\ref{lem:ghostly} (since the configuration is non-colliding), so we can now assume there are no ghost names. But if no agent has become triggered, then after $O(T_H)$ time we will reach a configuration $C_0$ where every remaining edge that was originally present in $C$ has $\edge.\timer = 0$. 

Now we wish to show that 
$C_0$ is safe.
Consider a coupled configuration $C_\text{empty}$, which is identical to $C_0$, except every tree only contains the root, 
i.e., is in the state just after executing \resetprotocol. 
Let $D'$ be the configuration where we follow the exact same execution that reached $D$ starting from configuration $C_0$.
We start by arguing that $C_\text{empty}$ is safe, then use this to argue that $C_0$ is safe.

Note that $C_\text{empty}$ is reachable from an awakening and non-colliding configuration (if all agents immediately executed \resetprotocol\ before future interactions). 
Thus we can apply Lemma~\ref{lem:safety-from-woke-configurations} to show \detectcollision\ will return \false\ from every interaction possible from $D'$, 
i.e., $C_\text{empty}$ is safe.
It remains to argue that $C_0$ is safe.

Now consider some interaction between agents $i$ and $j$ in configuration $D$, where they then instantiate \pathconsistencyprotocol\ on some path $P$ in $i.\tree$. In this case, every edge in $P$ has $\edge.\timer > 0$, so these all correspond to interactions that happened in the execution sequence from $C_0$. Because \pathconsistencyprotocol\ returns \true\ in the coupled configuration $D'$, there must be some matching reverse path in $j.\tree$. This same path will also be present in $D$, so \pathconsistencyprotocol\ also returns \true,
whence $C_0$ is safe.
\end{proof}

\begin{lemmarep}[Fast Detection]
\label{lem:fast-collision-detection}
If the configuration is colliding, then \detectcollision\ returns \true\ for some pair of agents in expected $O(T_H)$ time. When $H=\Theta(\log n)$, it also takes $O(T_H)$ with high probability $1-O(1/n)$.
\end{lemmarep}

\begin{proof}
Let $a, a'$ be two agents with $a.\name = a'.\name$. Consider the bounded epidemic process with agent $a$ as the source and $a'$ as the target, and time $\tau_{H+1}$ as defined in Lemma~\ref{lem:bounded-epidemic-constant}. At exactly time $\tau_{H+1}$, there will first be a sequence of interacting agents $P = (a,x_1,x_2,\ldots,x_h,a')$ with $h \leq H$, where $a$ interacted with $x_1$, then $x_1$ interacted with $x_2$, etc. 
Let $b = x_h$ be the agent that interacted with $a'$.
In the case where $H=O(1)$, Lemma~\ref{lem:bounded-epidemic-constant} gives $\IE[\tau_{H+1}] = O(Hn^{1/(H+1)}) = O(T_H)$, so $\tau_{H+1} = O(Hn^{1/(H+1)})$ with constant probability by Markov's inequality. In the case where $H = \Theta(\log n)$, we use Lemma~\ref{lem:bounded-epidemic-log-n} to instead say $\tau_{H+1} = O(\log n) = O(T_H)$ with high probability $1-O(1/n^2)$.

Let $b$ be the agent that interacts with $a'$ on this path. Then there is a sequence of interacting agents $Q = (a,x_1,\ldots,x_{h-1},b)$, ie. the prefix of $P$ without the final agent $a'$. There are $h \leq H$ agents before $b$, which is at most the maximum depth of $b.\tree$.
Because of this history of interactions, 
$b.\tree$ will contain a node labelled $a$ at depth $h$, whose edges have sync values corresponding to the various interactions in $Q$,
i.e., $b$'s tree has a root-to-leaf path $Q_r = (b,x_{h-1},\ldots,x_1,a)$ that is the reverse of $Q$.

We will now justify that all edges in $Q_r$ have $\edge.\timer > 0$, which will follow because we have only waited $O(T_H)$ time.
As the sequence of interacting agents $P$ grows from $a$ to $a'$, we consider the number of interactions from the agent at the current front of the sequence. This will give the number of times the edges on the eventual path $Q_r$ in $b.\tree$ have decremented $\edge.\timer$.
By standard Chernoff bounds, there will be at most $O(T_H)$ of these interactions within $O(T_H)$ parallel time with very high probability $1 - \exp(-\Theta(n))$. Since the initial value of $\edge.\timer = T_H$, it follows that all edges for path $P$ in $b.\tree$ still have positive timers. Thus by line~\ref{line:check-positive-timers} in \detectcollision, in the interaction between $b$ and $a'$, they will instantiate \pathconsistencyprotocol\ with this path $Q_r$ in $b.\tree$.

$a'$ has not interacted with any of the agents $a,x_1,\ldots,x_{h-1}$, so it has not had any interactions where it would learn any sync values of the path. Thus for each edge, the probability of $a'$ having a matching sync value is at most $\frac{1}{\synccount} = O\left( \frac{1}{n^2} \right)$. Then taking the union bound over the whole path of length at most $O(\log n)$, $a'$ does not have any matching sync values with high probability $1-O(\log n / n^2)$. So with high probability, \pathconsistencyprotocol\ returns \inconsistent, and \detectcollision\ returns \true\ as desired.

In the case where $H = O(1)$, the probability of all required events was at least some constant. Thus, we repeat the argument until \detectcollision\ returns \true, an expected constant number of times, to conclude that \detectcollision\ returns \true\ in expected $O(T_H)$ time. In the case where $H=\Theta(\log n)$, we have stronger high probability $1-O(1/n)$ for all required events, so we can in addition conclude that \detectcollision\ returns \true\ in $O(T_H)$ time with high probability $1-O(1/n)$.
\end{proof}
\end{toappendix}

Section~\ref{sec:detect-collision-proofs} proves that \detectcollision\ works in $O(T_H)$ time, and also satisfies required safety conditions that ensure there are no ``false positives'' where collisions are detected from configurations with unique names. These results will let us prove the main theorem about the behavior of \logTimeProtocol:

\begin{theoremrep}
\label{thm:log-time}
\logTimeProtocol\ uses $\exp(O(n^H) \log n)$ states.

When $H = O(1)$,
\logTimeProtocol\ solves self-stabilizing ranking in expected $O(H\cdot n^{1/(H+1)})$ time, and $O(H\cdot \log n \cdot n^{1/(H+1)})$ time with high probability $1-O(1/n)$.

When $H = \Theta(\log n)$,
\logTimeProtocol\ solves self-stabilizing ranking in time $O(\log n)$, in expectation and with high probability $1-O(1/n)$.
\end{theoremrep}

\begin{proof}
We first count the number of states by counting the number of bits each agent must store. The main memory cost comes from the field $\tree$, which has depth $H$, and each node can have at most $n$ children, so will have $O(n^H)$ nodes. Each node uses $O(\log n)$ bits for $\node.\name$, and each edge uses $O(\log n)$ bits for $\edge.\sync$ and $\edge.\timer$. Thus the tree uses $O(n^H \log n)$ bits. Compared to this, all other fields in the protocol are negligible, so it follows that the protocol uses $O(n^H \log n)$ bits, or $\exp(O(n^H) \log n)$ states.

We now argue that \logTimeProtocol\ solves self-stabilizing ranking. By Corollary~\ref{cor:reset-fully-computing}, we only have to consider initial configurations that are fully computing or awakening. We first consider fully computing configurations that are colliding. Here by Lemma~\ref{lem:fast-collision-detection}, the configuration becomes triggered in $O(T_H)$ time. When $H = O(1)$, this is in expectation (and thus with constant probability by Markov's inequality). When $H = \Theta(\log n)$, this is with high probability $1-O(1/n)$. We next consider fully computing configurations which are non-colliding. By Lemma~\ref{lem:safety-from-all-configurations}, after $O(T_H)$ time we reach either a partially triggered configuration, or a safe configuration. In the case where we reach a safe configuration, there were no ghost names, so every agent has $|\Met| = n$ and will have unique ranks based on the lexicographic ordering of $\Met$. Thus we have a stable ranked configuration.

We finally argue about awakening configurations. By Lemma~\ref{lem:unique-names}, this awakening configuration is non-colliding with high probability $1-O(1/n)$. In this case, by Lemma~\ref{lem:unique-ranks-from-unique-names}, we reach a configuration with unique ranks in $O(\log n)$ time, which is again a stable ranked configuration, since by Lemma~\ref{lem:safety-from-woke-configurations} this configuration is also safe.

For the time bounds, when $H = \Theta(\log n)$, we have $T_H = O(\log n)$, and all probability bounds were high probability $1-O(1/n)$, so we use $O(\log n)$ time in expectation and with high probability.

When $H = O(1)$, we have $T_H = O(H\cdot n^{1/(H+1)})$, and we stabilize in $O(T_H)$ time with constant probability. Then we can consider ``epochs'' of $O(T_H)$ time, where we declare an epoch successful if it stabilizes. Each epoch has a constant probability of being successful, so the expected number of required epochs is constant, giving an expected $O(H\cdot n^{1/(H+1)})$. With high probability $1-O(1/n)$, the required number of epochs is $O(\log n)$, giving time $O(H\cdot \log n \cdot n^{1/(H+1)})$ with high probability.
\end{proof}

%% file: Sections/figures/path1.tex
\begin{tabular}{|c|c|c|c|}
    \hline
     $a$'s tree & $b$'s tree & $c$'s tree & $d$'s tree  
     \\ \hline
     \rootoftree{a} & \rootoftree{b} & \rootoftree{c} & \rootoftree{d}
     \\ \hline 
     \multicolumn{4}{l}{$a$-$b$ interact; generate sync value $1$:}
     \\ \hline 
\begin{forest}
for tree={circle,draw, l sep=20pt}
[a
    [b, edge label={node[midway,left] {\newinfo{1}}}    ]
]
\end{forest}
& 
\begin{forest}
for tree={circle,draw, l sep=20pt}
[b
    [a, edge label={node[midway,left] {\newinfo{1}}}    ]
]
\end{forest}
& \rootoftree{c} & \rootoftree{d}
     \\ \hline 
     \multicolumn{4}{l}{$b$-$c$ interact; generate sync value $2$:}
     \\ \hline 
\begin{forest}
for tree={circle,draw, l sep=20pt}
[a
    [b, edge label={node[midway,left] {1}}    ]
]
\end{forest}
& 
\begin{forest}
for tree={circle,draw, l sep=20pt}
[b
    [a, edge label={node[midway,left] {1}}  ] 
    [c, edge label={node[midway,right] {\newinfo{2}}}    ]
]
\end{forest}
 & 
\begin{forest}
for tree={circle,draw, l sep=20pt}
[c
    [b, edge label={node[midway,left] {\newinfo{2}}}    
        [a, edge label={node[midway,left] {\newinfo{1}}}]
    ]
]
\end{forest}
 & \rootoftree{d}
     \\ \hline 
     \multicolumn{4}{l}{$c$-$d$ interact; generate sync value $3$:}
     \\ \hline 
\begin{forest}
for tree={circle,draw, l sep=20pt}
[a
    [b, edge label={node[midway,left] {1}}    ]
]
\end{forest}
& 
\begin{forest}
for tree={circle,draw, l sep=20pt}
[b
    [a, edge label={node[midway,left] {1}}  ] 
    [c, edge label={node[midway,right] {2}}    ]
]
\end{forest}
& 
\begin{forest}
for tree={circle,draw, l sep=20pt}
[c
    [b, edge label={node[midway,left] {2}}    
        [a, edge label={node[midway,left] {1}}]
    ]
    [d, edge label={node[midway,right] {\newinfo{3}}}]
]
\end{forest}
&
\begin{forest}
for tree={circle,draw, l sep=20pt}
[d
  [c, edge label={node[midway,left] {\newinfo{3}}}
    [b, edge label={node[midway,left] {\newinfo{2}}}
        [a, edge label={node[midway,left] {\newinfo{1}}}  ] 
    ]
]]
\end{forest}
     \\ \hline 
\end{tabular}

%% file: Sections/figures/path2.tex
\begin{tabular}{|c|c|c|c|}
    \hline
     $a$'s tree & $b$'s tree & $c$'s tree & $d$'s tree  
     \\ \hline
     \rootoftree{a} & \rootoftree{b} & \rootoftree{c} & \rootoftree{d}
     \\ \hline 
     \multicolumn{4}{l}{$a$-$b$ interact; generate sync value $1$:}
     \\ \hline 
\begin{forest}
for tree={circle,draw, l sep=20pt}
[a
    [b, edge label={node[midway,left] {\newinfo{1}}}    ]
]
\end{forest}
& 
\begin{forest}
for tree={circle,draw, l sep=20pt}
[b
    [a, edge label={node[midway,left] {\newinfo{1}}}    ]
]
\end{forest}
& \rootoftree{c} & \rootoftree{d}
     \\ \hline 
     \multicolumn{4}{l}{$b$-$c$ interact; generate sync value $2$:}
     \\ \hline 
\begin{forest}
for tree={circle,draw, l sep=20pt}
[a
    [b, edge label={node[midway,left] {1}}    ]
]
\end{forest}
& 
\begin{forest}
for tree={circle,draw, l sep=20pt}
[b
    [a, edge label={node[midway,left] {1}}  ] 
    [c, edge label={node[midway,right] {\newinfo{2}}}    ]
]
\end{forest}
 & 
\begin{forest}
for tree={circle,draw, l sep=20pt}
[c
    [b, edge label={node[midway,left] {\newinfo{2}}}    
        [a, edge label={node[midway,left] {\newinfo{1}}}]
    ]
]
\end{forest}
 & \rootoftree{d}
     \\ \hline 
     \multicolumn{4}{l}{$a$-$b$ interact; generate sync value $7$:}
     \\ \hline 
\begin{forest}
for tree={circle,draw, l sep=20pt}
[a
    [b, edge label={node[midway,left] {\newinfo{7}}}    
            [c, edge label={node[midway,left] {\newinfo{2}}}    ]
    ]
]
\end{forest}
& 
\begin{forest}
for tree={circle,draw, l sep=20pt}
[b
    [a, edge label={node[midway,left] {\newinfo{7}}}  ] 
    [c, edge label={node[midway,right] {2}}    ]
]
\end{forest} & 
\begin{forest}
for tree={circle,draw, l sep=20pt}
[c
    [b, edge label={node[midway,left] {2}}    
        [a, edge label={node[midway,left] {1}}]
    ]
]
\end{forest}
& \rootoftree{d}
     \\ \hline 
     \multicolumn{4}{l}{$c$-$d$ interact; generate sync value $3$:}
     \\ \hline 
\begin{forest}
for tree={circle,draw, l sep=20pt}
[a
    [b, edge label={node[midway,left] {7}}    
        [c, edge label={node[midway,left] {2}}    ]
    ]
]
\end{forest}
& 
\begin{forest}
for tree={circle,draw, l sep=20pt}
[b
    [a, edge label={node[midway,left] {7}}  ] 
    [c, edge label={node[midway,right] {2}}    ]
]
\end{forest}
& 
\begin{forest}
for tree={circle,draw, l sep=20pt}
[c
    [b, edge label={node[midway,left] {2}}    
        [a, edge label={node[midway,left] {1}}]
    ]
    [d, edge label={node[midway,right] {\newinfo{3}}}]
]
\end{forest}
&
\begin{forest}
for tree={circle,draw, l sep=20pt}
[d
  [c, edge label={node[midway,left] {\newinfo{3}}}
    [b, edge label={node[midway,left] {\newinfo{2}}}
        [a, edge label={node[midway,left] {\newinfo{1}}}  ] 
    ]
]]
\end{forest}
     \\ \hline 
\end{tabular}

%% file: Sections/synthetic-coin.tex
\begin{toappendix}

\section{Derandomization of Protocols}\label{sec:derandomization}

\newcommand{\alg}{\mathsf{Alg}}
\newcommand{\f}{\mathsf{Flip}}

Note that our model as defined allowed random transitions. However, this was simply for ease of presentation, and the randomness can be simulated through standard ``synthetic coin'' techniques that exploit the randomness of the scheduler.

We only used randomness in the \resetprotocol\ for \logTimeProtocol, to generate a name uniformly from the set $\{0,1\}^{3 \log^2 n}$.
We show one approach for how an agent can collect $O(\log n)$ random bits to generate this name.

This approach was inspired by a similar technique due to 
Sudo, Ooshita, Izumi, Kakugawa, and Masuzawa~\cite{DBLP:journals/corr/abs-1812-11309},
but substitutes their ``space multiplexing'' 
(splitting the population into two approximately equal-size subpopulations $A$ and $F$,
which is not clear how to implement in a self-stabilizing manner)
with ``time-multiplexing''.
On each interaction the agent switches between two roles:
``normal algorithm'' role ($\alg$), 
and ``coin flip'' role ($\f$).
When an agent needs a random bit, 
it waits until it is role $\alg$ and its partner is role $\f$.
If $\alg$ is the initiator, this represents heads, and if $\alg$ is the responder, this represents tails.
This decouples any dependence of the coin flips on each other or on the state of the agent being interacted with. 
It also incurs an expected slowdown of factor only $1/4$ per bit 
(since each agent requiring a random bit waits expected 4 interactions until it is in role $\alg$ and the other is in role $\f$). 
Thus, by the Chernoff bound, with high probability, the actual slowdown over all $O(\log n)$ bits is at most of factor $1/8$. 

Thus the agents can become inactive during \resetprotocol\ for the $O(\log n)$ interactions it takes to generate enough random bits to create a new name.

Our constructions here relied on initiator / responder asymmetry as the source of randomness. 
It would also be possible to avoid this capability and use a symmetric synthetic coin technique as described in \cite{DBLP:conf/soda/AlistarhAEGR17}.
In that case, our protocols would be almost entirely symmetric, with one exception: 
the slow leader election $L,L \to L,F$ required as part of \resetprotocol\ used in \silentlinearTimeStateProtocol.
This line itself could use the symmetric synthetic-coin to be simulated with symmetric transitions, leading to entirely symmetric protocols.

\end{toappendix}

%% file: Sections/conclusion.tex
\section{Conclusion and Perspectives}
\label{sec:conclusion}

For the first time, we addressed time-space trade-offs of self-stabilizing leader election and ranking in population protocols over complete graphs. 
We emphasize that solving these problems, while ensuring such a strong form of fault-tolerance, 
necessitates linear states and strong nonuniformity (Theorem \ref{thm:n-state-lower-bound}). 
Other forms of ``strong’’ fault-tolerance, such as 
Byzantine-tolerance
\cite{DBLP:conf/icalp/GuerraouiR09} or loosely-stabilizing leader election with \emph{exponential holding time} 
(a period of time where a unique leader persists after stabilization) \cite{DBLP:conf/sirocco/Izumi15,DBLP:journals/tcs/SudoOKMDL20}, similarly necessitate $\Omega(n)$ states. 
By contrast, a sublinear number of states suffices for many non-fault-tolerant protocols (cf.~\cite{DBLP:journals/sigact/AlistarhG18}) and weaker forms of tolerance, 
such as loosely-stabilizing leader election with \emph{polynomial} holding time \cite{DBLP:journals/tcs/SudoOKMDL20} or tolerance to a \emph{constant} number of crashes and transient faults \cite{DBLP:conf/dcoss/Delporte-GalletFGR06}.\footnote{
Recall that self-stabilization tolerates any number of transient faults.
}


To conclude, we propose several perspectives.


\ifarticle
    \paragraph{Time/space tradeoffs.}
\else
    \noindent{\textbf{Time/space tradeoffs.}} 
\fi
It is open to find a subexponential-state sublinear-time self-stabilizing ranking protocol.
Observation~\ref{obs:linear-time-lower-bound-silent} states that any sublinear time SSR protocol is not silent.
\logTimeProtocol\ is non-silent because it perpetually passes around information about agents' recent interactions with each other,
as a way to detect name collisions without requiring the agents with equal names to meet directly.
Even when limiting the tree of interactions to depth 1, 
this results in an exponential number of states,
since each agent must maintain a value to associate to every other agent in the population.
Thus, a subexponential-state protocol 
(if based upon fast collision detection) 
would somehow need to embed enough information in each agent to enable fast collision detection,
while somehow allowing the agent to forget ``most'' of the information about its interactions. 
Furthermore, our strategy of using the set $\Met$ of all names to go from unique names to unique ranks fundamentally requires exponential states.

\ifarticle
    \paragraph{Ranking vs.~leader election.}
\else
    \noindent{\textbf{Ranking vs.~leader election.}} 
\fi
Ranking implies leader election (``automatically''), 
but the converse does not hold.
In the \emph{initialized} case where we can specify an initial state for each agent,
it is possible to elect a leader without ranking, 
using the single transition $\ell,\ell \to \ell,f$
(using too few states for the ranking problem even to be definable).
Though any \emph{self-stabilizing} protocol for leader election must use at least $n$ states~\cite{cai2012prove} (Theorem \ref{thm:n-state-lower-bound} here),
it is not the case that any SSLE protocol implicitly 
solves the ranking problem. 
\opt{inline,append}{
    (See Observation~\ref{obs:leader-election-without-ranking}.)
}
It would be interesting to discover an SSLE algorithm that is more efficient than our examples because it does not also solve ranking.

\ifarticle
    \paragraph{Initialized ranking.}
\else
    \noindent{\textbf{Initialized ranking.}} 
\fi
In the other direction, consider the ranking problem in a non-self-stabilizing setting. Without the constraint of self-stabilization, there is no longer the issue of ghost names.
Compared to self-stabilization,
it may be easier to find an \emph{initialized} ranking protocol that still uses polylogarithmic time, 
but only polynomial states.

\ifarticle
    \paragraph{Initialized collision detection.}
\else
    \noindent{\textbf{Initialized collision detection.}} 
\fi
The core difficulty of \logTimeProtocol\ is \emph{collision detection}:
discovering that two agents have been assigned the same name without waiting $\Theta(n)$ time for them to interact directly.
It would be interesting to study this problem in the (non-self-stabilizing) setting where an adversary assigns read-only names to each agent,
but the read/write memory can be initialized to the same state for each agent.
Can a name collision be detected in sublinear time and sub-exponential states?



\newcommand{\h}{\texttt{h}}
\newcommand{\tail}{\texttt{t}}

%% file: Selfstabilizing.bbl
\begin{thebibliography}{10}

\bibitem{DBLP:conf/soda/AlistarhAEGR17}
D.~Alistarh, J.~Aspnes, D.~Eisenstat, R.~Gelashvili, and R.~L. Rivest.
\newblock Time-space trade-offs in population protocols.
\newblock In {\em SODA}, pages 2560--2579, 2017.

\bibitem{DBLP:conf/soda/AlistarhAG18}
D.~Alistarh, J.~Aspnes, and R.~Gelashvili.
\newblock Space-optimal majority in population protocols.
\newblock In {\em SODA}, pages 2221--2239, 2018.

\bibitem{DBLP:conf/dna/AlistarhDKSU17}
D.~Alistarh, B.~Dudek, A.~Kosowski, D.~Soloveichik, and P.~Uznanski.
\newblock Robust detection in leak-prone population protocols.
\newblock In {\em {DNA}}, pages 155--171, 2017.

\bibitem{DBLP:journals/sigact/AlistarhG18}
D.~Alistarh and R.~Gelashvili.
\newblock Recent algorithmic advances in population protocols.
\newblock {\em {SIGACT} News}, 49(3):63--73, 2018.

\bibitem{alistarh2010fast}
Dan Alistarh, Hagit Attiya, Seth Gilbert, Andrei Giurgiu, and Rachid Guerraoui.
\newblock Fast randomized test-and-set and renaming.
\newblock In {\em DISC 2010: International Symposium on Distributed Computing},
  pages 94--108. Springer, 2010.

\bibitem{alistarh2014balls}
Dan Alistarh, Oksana Denysyuk, Lu{\'\i}s Rodrigues, and Nir Shavit.
\newblock Balls-into-leaves: Sub-logarithmic renaming in synchronous
  message-passing systems.
\newblock In {\em PODC 2014: Proceedings of the 2014 ACM Symposium on
  Principles of Distributed Computing}, pages 232--241, 2014.

\bibitem{amir2020message}
Talley Amir, James Aspnes, David Doty, Mahsa Eftekhari, and Eric Severson.
\newblock Message complexity of population protocols.
\newblock In Hagit Attiya, editor, {\em DISC 2020: 34th International Symposium
  on Distributed Computing}, volume 179 of {\em Leibniz International
  Proceedings in Informatics (LIPIcs)}, pages 6:1--6:18, Dagstuhl, Germany,
  2020. Schloss Dagstuhl--Leibniz-Zentrum f{\"u}r Informatik.

\bibitem{DBLP:journals/dc/AngluinADFP06}
D.~Angluin, J.~Aspnes, Z.~Diamadi, M.~J. Fischer, and R.~Peralta.
\newblock Computation in networks of passively mobile finite-state sensors.
\newblock {\em Distributed Computing}, 18(4):235--253, 2006.

\bibitem{DBLP:journals/dc/AngluinAE08a}
D.~Angluin, J.~Aspnes, and D.~Eisenstat.
\newblock Fast computation by population protocols with a leader.
\newblock {\em Distributed Computing}, 21(3):183--199, 2008.

\bibitem{DBLP:conf/opodis/AngluinAFJ05}
D.~Angluin, J.~Aspnes, M.~J. Fischer, and H.~Jiang.
\newblock Self-stabilizing population protocols.
\newblock In {\em {OPODIS}}, volume 3974, pages 103--117. Springer, 2005.

\bibitem{DBLP:conf/opodis/BeauquierBB13}
J.~Beauquier, P.~Blanchard, and J.~Burman.
\newblock Self-stabilizing leader election in population protocols over
  arbitrary communication graphs.
\newblock In {\em OPODIS}, pages 38--52, 2013.

\bibitem{Beauquier2007}
J.~Beauquier, J.~Clement, S.~Messika, L.~Rosaz, and B.~Rozoy.
\newblock Self-stabilizing counting in mobile sensor networks with a base
  station.
\newblock In {\em DISC}, pages 63--76, 2007.

\bibitem{DBLP:conf/icalp/BellevilleDS17}
A.~Belleville, D.~Doty, and D.~Soloveichik.
\newblock Hardness of computing and approximating predicates and functions with
  leaderless population protocols.
\newblock In {\em {ICALP}}, pages 141:1--141:14, 2017.

\bibitem{DBLP:conf/stoc/BerenbrinkGK20}
P.~Berenbrink, G.~Giakkoupis, and P.~Kling.
\newblock Optimal time and space leader election in population protocols.
\newblock In {\em {STOC}}, pages 119--129. {ACM}, 2020.

\bibitem{DBLP:conf/stacs/BlondinEJ18}
M.~Blondin, J.~Esparza, and S.~Jaax.
\newblock Large flocks of small birds: on the minimal size of population
  protocols.
\newblock In {\em {STACS}}, pages 16:1--16:14, 2018.

\bibitem{DBLP:journals/corr/abs-0906-3256}
O.~Bournez, J.~Chalopin, J.~Cohen, and X.~Koegler.
\newblock Playing with population protocols.
\newblock In {\em {CSP} 2008}, pages 3--15, 2008.

\bibitem{DBLP:journals/mst/BournezCR18}
O.~Bournez, J.~Cohen, and M.~Rabie.
\newblock Homonym population protocols.
\newblock {\em Theory of Computing Systems}, 62(5):1318--1346, 2018.

\bibitem{BowerB04}
J.~M. Bower and H.~Bolouri.
\newblock {\em Computational modeling of genetic and biochemical networks.}
\newblock MIT press, 2004.

\bibitem{RandomExchangesSteele}
D.~W. Boyd and J.~M. Steele.
\newblock Random exchanges of information.
\newblock {\em Journal of Applied Probability}, 16(3):657--661, 1979.

\bibitem{DBLP:conf/wdag/BurmanBS19}
J.~Burman, J.~Beauquier, and D.~Sohier.
\newblock Space-optimal naming in population protocols.
\newblock In J.~Suomela, editor, {\em DISC'19}, volume 146, pages 9:1--9:16,
  2019.

\bibitem{cai2012prove}
S.~Cai, T.~Izumi, and K.~Wada.
\newblock How to prove impossibility under global fairness: On space complexity
  of self-stabilizing leader election on a population protocol model.
\newblock {\em Theory of Computing Systems}, 50(3):433--445, 2012.

\bibitem{casteigts2020sharp}
A.~Casteigts, M.~Raskin, M.~Renken, and V.~Zamaraev.
\newblock Sharp thresholds in random simple temporal graphs, 2020.

\bibitem{chalk2021composable}
Cameron Chalk, Niels Kornerup, Wyatt Reeves, and David Soloveichik.
\newblock Composable rate-independent computation in continuous chemical
  reaction networks.
\newblock {\em IEEE/ACM Transactions on Computational Biology and
  Bioinformatics}, 18(1):250--260, 2021.
\newblock special issue of invited papers from CMSB 2018.

\bibitem{chen2019self}
H.-P. Chen and H.-L. Chen.
\newblock Self-stabilizing leader election.
\newblock In {\em Proceedings of the 2019 ACM Symposium on Principles of
  Distributed Computing}, PODC '19, page 53–59, New York, NY, USA, 2019.
  Association for Computing Machinery.

\bibitem{chen2020self}
H.-P. Chen and H.-L. Chen.
\newblock Self-stabilizing leader election in regular graphs.
\newblock In {\em Proceedings of the 39th Symposium on Principles of
  Distributed Computing}, PODC '20, page 210–217, New York, NY, USA, 2020.
  Association for Computing Machinery.

\bibitem{DBLP:conf/dcoss/Delporte-GalletFGR06}
C.~Delporte-Gallet, H.~Fauconnier, R.~Guerraoui, and E.~Ruppert.
\newblock When birds die: Making population protocols fault-tolerant.
\newblock In {\em DCOSS}, pages 51--66, 2006.

\bibitem{Diamadi2001}
Z.~Diamadi and M.~J. Fischer.
\newblock A simple game for the study of trust in distributed systems.
\newblock {\em Wuhan University Journal of Natural Sciences}, 6(1):72--82, Mar
  2001.

\bibitem{dijkstra}
E.~W. Dijkstra.
\newblock Self-stabilizing systems in spite of distributed control.
\newblock {\em Commun. of the ACM}, 17(11):643--644, Nov. 1974.

\bibitem{dolev2000self}
Shlomi Dolev.
\newblock {\em Self-stabilization}.
\newblock MIT press, 2000.

\bibitem{DBLP:journals/dc/DolevIM93}
Shlomi Dolev, Amos Israeli, and Shlomo Moran.
\newblock Self-stabilization of dynamic systems assuming only read/write
  atomicity.
\newblock {\em Distributed Comput.}, 7(1):3--16, 1993.

\bibitem{DBLP:conf/wdag/DotyS15}
D.~Doty and D.~Soloveichik.
\newblock Stable leader election in population protocols requires linear time.
\newblock In {\em DISC}, pages 602--616, 2015.

\bibitem{drmota2009height}
M.~Drmota.
\newblock The height of increasing trees.
\newblock {\em Annals of Combinatorics}, 12(4):373--402, 2009.

\bibitem{Drmota09}
M.~Drmota.
\newblock {\em Random Trees: An Interplay between Combinatorics and
  Probability}.
\newblock Springer, Heidelberg, 2009.

\bibitem{DBLP:journals/eatcs/ElsasserR18}
R.~Els{\"{a}}sser and T.~Radzik.
\newblock Recent results in population protocols for exact majority and leader
  election.
\newblock {\em Bulletin of the {EATCS}}, 126, 2018.

\bibitem{DBLP:journals/acta/EsparzaGLM17}
J.~Esparza, P.~Ganty, J.~Leroux, and R.~Majumdar.
\newblock Verification of population protocols.
\newblock {\em Acta Informatica}, 54(2):191--215, 2017.

\bibitem{DBLP:conf/opodis/FischerJ06}
M.~J. Fischer and H.~Jiang.
\newblock Self-stabilizing leader election in networks of finite-state
  anonymous agents.
\newblock In {\em OPODIS}, pages 395--409, 2006.

\bibitem{DBLP:conf/soda/GasieniecS18}
L.~Gasieniec and G.~Stachowiak.
\newblock Fast space optimal leader election in population protocols.
\newblock In {\em {SODA}}, pages 2653--2667, 2018.

\bibitem{DBLP:conf/spaa/GasieniecSU19}
L.~Gasieniec, G.~Stachowiak, and P.~Uznanski.
\newblock Almost logarithmic-time space optimal leader election in population
  protocols.
\newblock In {\em SPAA}, pages 93--102, 2019.

\bibitem{Gillespie1977}
D.~T. Gillespie.
\newblock Exact stochastic simulation of coupled chemical reactions.
\newblock {\em Journal of Physical Chemistry}, 81 (25):2340 -- 2361, 1977.

\bibitem{DBLP:conf/icalp/GuerraouiR09}
R.~Guerraoui and E.~Ruppert.
\newblock Names trump malice: Tiny mobile agents can tolerate byzantine
  failures.
\newblock In {\em ICALP (2)}, pages 484--495, 2009.

\bibitem{RandomExchangesHaigh}
J.~Haigh.
\newblock Random exchanges of information.
\newblock {\em Journal of Applied Probability}, 18(3):743--746, 1981.

\bibitem{DBLP:conf/sirocco/Izumi15}
T.~Izumi.
\newblock On space and time complexity of loosely-stabilizing leader election.
\newblock In {\em {SIROCCO}}, volume 9439 of {\em Lecture Notes in Computer
  Science}, pages 299--312. Springer, 2015.

\bibitem{janson2018tail}
S.~Janson.
\newblock Tail bounds for sums of geometric and exponential variables.
\newblock {\em Statistics \& Probability Letters}, 135:1--6, 2018.

\bibitem{DBLP:conf/infocom/JohnsonSFFSRL06}
D.~Johnson, T.~Stack, R.~Fish, D.~Montrallo Flickinger, L.~Stoller, R.~Ricci,
  and J.~Lepreau.
\newblock Mobile emulab: {A} robotic wireless and sensor network testbed.
\newblock In {\em {INFOCOM}}. {IEEE}, 2006.

\bibitem{DBLP:conf/podc/KosowskiU18}
A.~Kosowski and P.~Uznanski.
\newblock Brief announcement: Population protocols are fast.
\newblock In {\em {PODC}}, pages 475--477, 2018.

\bibitem{DBLP:journals/tcs/LunaFIISV19}
G.~Di Luna, P.~Flocchini, T.~Izumi, T.~Izumi, N.~Santoro, and G.~Viglietta.
\newblock Population protocols with faulty interactions: The impact of a
  leader.
\newblock {\em Theoretical Computer Science}, 754:35--49, 2019.

\bibitem{doi:10.2200/S00328ED1V01Y201101DCT006}
O.~Michail, I.~Chatzigiannakis, and P.~G. Spirakis.
\newblock New models for population protocols.
\newblock {\em Synthesis Lectures on Distributed Computing Theory},
  2(1):1--156, 2011.

\bibitem{DBLP:conf/sss/MichailCS13}
O.~Michail, I.~Chatzigiannakis, and P.~G. Spirakis.
\newblock Naming and counting in anonymous unknown dynamic networks.
\newblock In {\em SSS}, pages 281--295, 2013.

\bibitem{DBLP:journals/dc/MizoguchiOKY12}
R.~Mizoguchi, H.~Ono, S.~Kijima, and M.~Yamashita.
\newblock On space complexity of self-stabilizing leader election in mediated
  population protocol.
\newblock {\em Distributed Computing}, 25(6):451--460, 2012.

\bibitem{mocquard2016analysis}
Y.~Mocquard, B.~Sericola, S.~Robert, and E.~Anceaume.
\newblock Analysis of the propagation time of a rumour in large-scale
  distributed systems.
\newblock In {\em 2016 IEEE 15th International Symposium on Network Computing
  and Applications (NCA)}, pages 264--271. IEEE, 2016.

\bibitem{RandomExchangesMoon}
J.~W. Moon.
\newblock Random exchanges of information.
\newblock {\em Nieuw Archief voor Wiskunde}, 20:246–--249, 1972.

\bibitem{DBLP:conf/sensys/PolastreHC04}
J.~Polastre, J.~L. Hill, and D.~E. Culler.
\newblock Versatile low power media access for wireless sensor networks.
\newblock In {\em SenSys}, pages 95--107. {ACM}, 2004.

\bibitem{DBLP:conf/sirocco/Rabie17}
M.~Rabie.
\newblock Global versus local computations: Fast computing with identifiers.
\newblock In {\em SIROCCO}, pages 90--105, 2017.

\bibitem{severson2020composable}
Eric Severson, David Haley, and David Doty.
\newblock Composable computation in discrete chemical reaction networks.
\newblock {\em Distributed Computing}, 2020.
\newblock to appear. Special issue of invited papers from PODC 2019.

\bibitem{DBLP:journals/nc/SoloveichikCWB08}
D.~Soloveichik, M.~Cook, E.~Winfree, and J.~Bruck.
\newblock Computation with finite stochastic chemical reaction networks.
\newblock {\em Natural Computing}, 7(4):615--633, 2008.

\bibitem{sudo2020leader}
Y.~Sudo and T.~Masuzawa.
\newblock Leader election requires logarithmic time in population protocols.
\newblock {\em Parallel Processing Letters}, 30(01):2050005, 2020.

\bibitem{DBLP:journals/corr/abs-1812-11309}
Y.~Sudo, F.~Ooshita, T.~Izumi, H.~Kakugawa, and T.~Masuzawa.
\newblock Logarithmic expected-time leader election in population protocol
  model.
\newblock {\em CoRR}, abs/1812.11309, 2018.

\bibitem{sudo2020logarithmic}
Y.~{Sudo}, F.~{Ooshita}, T.~{Izumi}, H.~{Kakugawa}, and T.~{Masuzawa}.
\newblock Time-optimal leader election in population protocols.
\newblock {\em IEEE Transactions on Parallel and Distributed Systems},
  31(11):2620--2632, 2020.

\bibitem{DBLP:journals/tcs/SudoOKMDL20}
Y.~Sudo, F.~Ooshita, H.~Kakugawa, T.~Masuzawa, A.~K. Datta, and L.~L. Larmore.
\newblock Loosely-stabilizing leader election with polylogarithmic convergence
  time.
\newblock {\em Theor. Comput. Sci.}, 806:617--631, 2020.

\bibitem{DBLP:conf/sirocco/SudoS0KM20}
Y.~Sudo, M.~Shibata, J.~Nakamura, Y.~Kim, and T.~Masuzawa.
\newblock The power of global knowledge on self-stabilizing population
  protocols.
\newblock In {\em {SIROCCO}}, volume 12156, pages 237--254. Springer, 2020.

\bibitem{DBLP:conf/sss/XuYKY13}
X.~Xu, Y.~Yamauchi, S.~Kijima, and M.~Yamashita.
\newblock Space complexity of self-stabilizing leader election in population
  protocol based on k-interaction.
\newblock In {\em SSS}, pages 86--97, 2013.

\bibitem{DBLP:conf/opodis/YasumiOYI17}
H.~Yasumi, F.~Ooshita, K.~Yamaguchi, and M.~Inoue.
\newblock Constant-space population protocols for uniform bipartition.
\newblock In {\em {OPODIS} 2017}, pages 19:1--19:17, 2017.

\bibitem{DBLP:conf/sss/YokotaSM20}
D.~Yokota, Y.~Sudo, and T.~Masuzawa.
\newblock Time-optimal self-stabilizing leader election on rings in population
  protocols.
\newblock In {\em {SSS}}, volume 12514, pages 301--316. Springer, 2020.

\end{thebibliography}
